\documentclass[11pt,pagebackref,letterpaper=true,colorlinks=true,pdfpagemode=none,urlcolor=blue,linkcolor=blue,citecolor=blue,pdfstartview=FitH]{article}
\usepackage{graphicx}
\usepackage{algorithm,algorithmic}
\usepackage{amsmath, amssymb, amsthm}
\usepackage{url}
\usepackage{hyperref}
\usepackage{ifthen}
\usepackage{xcolor}
\usepackage{algorithmic,algorithm}
\usepackage[normalem]{ulem}
\usepackage{thmtools}
\usepackage{thm-restate}
\usepackage{times}
\usepackage{bm}
\usepackage{enumitem}
\setlist{nolistsep}

\usepackage{xspace}

\DeclareMathOperator{\rank}{rank}
\DeclareMathOperator{\maxrank}{maxrank}
\DeclareMathOperator{\mrank}{\maxrank}
\DeclareMathOperator{\argmax}{argmax}

\newtheorem{theorem}{Theorem}[section]

\newtheorem{lemma}[theorem]{Lemma}
\newtheorem{claim}[theorem]{Claim}

\declaretheorem[sibling=theorem, style=definition]{definition}

\newenvironment{proofof}[1]{\begin{proof}[Proof of #1]}{\end{proof}}

\newcommand{\ignore}[1]{}
\renewcommand{\cal}[1]{\mathcal #1}

\newcommand{\R}{\mathbb R}

\newcommand{\Z}{{\mathbb Z}}

\newcommand{\eps}{\varepsilon}

\newcommand{\cI}{{\cal I}}

\newcommand{\cL}{{\cal L}}
\newcommand{\cM}{{\cal M}}
\newcommand{\cN}{{\cal N}}

\newcommand{\cP}{{\cal P}}
\newcommand{\cQ}{{\cal Q}}

\newcommand{\bp}{\mathbf{p}}
\newcommand{\bq}{\mathbf{q}}

\newcommand{\ceil}[1]{\ensuremath{\left\lceil {#1} \right\rceil}}
\newcommand{\floor}[1]{\ensuremath{\left\lfloor {#1} \right\rfloor}}

\DeclareMathOperator{\Exp}{\mathbf{E}}

\renewcommand{\Pr}{\mathbf{Pr}}

\newcommand{\mySec}[1]{\hyperref[#1]{Section~\ref*{#1}}} 
\newcommand{\twoSecs}[2]{\hyperref[#1]{Sections~\ref*{#1}} and \hyperref[#2]{~\ref*{#2}}} 
\newcommand{\myApp}[1]{\hyperref[#1]{Appendix~\ref*{#1}}} 
\newcommand{\Sec}[1]{\hyperref[sec:#1]{Section~\ref*{sec:#1}}} 
\newcommand{\Eqn}[1]{\hyperref[eq:#1]{(\ref*{eq:#1})}} 
\newcommand{\Fig}[1]{\hyperref[fig:#1]{Fig.\,\ref*{fig:#1}}} 
\newcommand{\Tab}[1]{\hyperref[tab:#1]{Tab.\,\ref*{tab:#1}}} 
\newcommand{\Thm}[1]{\hyperref[thm:#1]{Theorem\,\ref*{thm:#1}}} 
\newcommand{\Lem}[1]{\hyperref[lem:#1]{Lemma\,\ref*{lem:#1}}} 
\newcommand{\Prop}[1]{\hyperref[prop:#1]{Prop.~\ref*{prop:#1}}} 
\newcommand{\Cor}[1]{\hyperref[cor:#1]{Corollary~\ref*{cor:#1}}} 
\newcommand{\Def}[1]{\hyperref[def:#1]{Definition~\ref*{def:#1}}} 
\newcommand{\Alg}[1]{\hyperref[alg:#1]{Alg.~\ref*{alg:#1}}} 
\newcommand{\Ex}[1]{\hyperref[ex:#1]{Ex.~\ref*{ex:#1}}} 
\newcommand{\Clm}[1]{\hyperref[clm:#1]{Claim~\ref*{clm:#1}}} 
\colorlet{shadecolor}{blue!20}

\def\bsucc{\succ}
\def\bsucceq{\succeq}

\newcommand{\sm}{\ensuremath{\setminus}}
\newcommand{\es}{\ensuremath{\emptyset}}
\newcommand{\sse}{\subseteq}

\newcommand{\al}{\ensuremath{\alpha}}
\newcommand{\sg}{\ensuremath{\sigma}}
\newcommand{\Sg}{\ensuremath{\Sigma}}

\newcommand{\Gm}{\ensuremath{\Gamma}}

\newcommand{\e}{\ensuremath{\epsilon}}
\newcommand{\Dt}{\ensuremath{\Delta}}
\newcommand{\ve}{\ensuremath{\varepsilon}}
\newcommand{\tx}{\ensuremath{\tilde x}}

\newcommand{\hi}{\ensuremath{\hat i}}
\newcommand{\ho}{\ensuremath{\hat o}}
\newcommand{\mat}{\ensuremath{M}}
\newcommand{\hist}{\ensuremath{\mathrm{hist}}}
\newcommand{\pos}[2]{\ensuremath{\mathrm{pos}({#1},{#2})}}
\newcommand{\alt}[2]{\ensuremath{\mathrm{alt}({#1},{#2})}}
\newcommand{\univt}{\ensuremath{\mathsf{UnivT}}\xspace}
\newcommand{\strongt}{\ensuremath{\mathsf{StrongT}}\xspace}
\newcommand{\weakt}{\ensuremath{\mathsf{WeakT}}\xspace}
\newcommand{\lext}{\ensuremath{\mathsf{LexT}}\xspace}
\newcommand{\mma}{\ensuremath{\mathsf{MaxMatch}}\xspace}
\newcommand{\pl}{\ensuremath{\mathsf{Pl}}}

\def\1{\mathbf 1}

\title{Welfare Maximization and Truthfulness in Mechanism Design
with Ordinal Preferences%
\footnote{Appeared in the Proceedings of the 5th Innovations in Theoretical Computer
  Science conference, 2014~\cite{ChakrabartyS14}.}}
\author{
         Deeparnab Chakrabarty\thanks{{\tt dechakr@microsoft.com.}
         Microsoft Research, India.} 
\and
         Chaitanya Swamy\thanks{{\tt cswamy@math.uwaterloo.ca}.
         Dept. of Combinatorics and Optimization, Univ. Waterloo, Waterloo, ON N2L 3G1.
         Supported in part by NSERC grant 327620-09, an NSERC Discovery Accelerator
         Supplement Award, and an Ontario Early Researcher Award.} 
}
\date{}

\begin{document}

\maketitle

\begin{abstract}
We study mechanism design problems in the {\em ordinal setting} wherein the
preferences of agents are described by orderings over outcomes, as opposed to specific
numerical values associated with them. This setting is relevant when agents can compare
outcomes, but aren't able to evaluate precise utilities for them. Such a situation arises
in diverse contexts including voting and matching markets. 

Our paper addresses two issues that arise in ordinal mechanism design. To design social
welfare maximizing mechanisms, one needs to be able to quantitatively measure the welfare
of an outcome which is not clear in the ordinal setting. Second, since the impossibility
results of Gibbard and Satterthwaite~\cite{Gibbard73,Satterthwaite75} force one to move to
randomized mechanisms, one needs a more nuanced notion of truthfulness. 

We propose {\em rank approximation} as a metric for measuring the quality of an
outcome, which allows us to evaluate mechanisms based on worst-case performance, 
and {\em lex-truthfulness} as a notion of truthfulness for randomized ordinal
mechanisms. Lex-truthfulness is stronger than notions studied in the literature, and yet
flexible enough to admit a rich class of mechanisms {\em circumventing classical
impossibility results}. 
We demonstrate the usefulness of the above notions by devising lex-truthful
mechanisms achieving good rank-approximation factors, both in the general ordinal
setting, as well as structured settings such as {\em (one-sided) matching markets}, and
its generalizations, {\em matroid} and {\em scheduling} markets.
\end{abstract}

\section{Introduction}\label{sec:intro} \label{intro}

A central problem in social choice theory and mechanism design is that of choosing a
``good'' outcome by aggregating individuals' private preferences over outcomes, where 
individuals are rational agents. 
A {\em mechanism} implementing a {\em social choice function} (SCF) needs to elicit the 
preferences of agents in a truthful fashion, that is, in a way such that no agent may
benefit by misreporting his preferences. 

In this paper, we study mechanism-design problems in {\em ordinal} settings, wherein the
preferences are described by orderings over the set of outcomes. 
This is in contrast with the {\em cardinal} setting, wherein an agent specifies a 
{\em value} to each outcome (which determines his preferences). 
Ordinal settings reduce the ``informational burden'' on an agent in the sense that he only
needs to be able to compare outcomes rather than assign values to outcomes
justifying his preferences. It is not hard to imagine settings where the former
comparison task is easier, and more aptly describes the situation: 
examples span the spectrum between electoral settings and the 
setting of allocating dormitory rooms to students. 

Two immediate issues arise in ordinal mechanism design. A typical mechanism-design goal is
to maximize {\em social welfare}, but in order to approach this goal in ordinal settings,
one needs to first be able to {\em quantitatively} measure the social-welfare value of an
outcome. Second, 
since the Gibbard-Satterthwaite (GS) impossibility result~\cite{Gibbard73,Satterthwaite75} 
precludes non-trivial deterministic truthful mechanisms, one is forced to move to
{\em randomized mechanisms} for which one needs a more nuanced notion of truthfulness.  

\subsection{Our contributions}
We propose a novel framework for welfare-maximization and truthfulness for randomized
ordinal mechanisms, and devise various near-optimal mechanisms in this framework. 
Our contributions are threefold.
\begin{list}{\arabic{enumi})}{\usecounter{enumi} \topsep=0.25ex \itemsep=0ex
    \addtolength{\leftmargin}{-3ex}}

\item 
We introduce a metric called {\em rank approximation} for measuring the quality of
an outcome, which in turn allows us to evaluate mechanisms in terms of their worst-case
performance. We show that rank approximation is a robust notion that is appealing,
and can be motivated, from various perspectives. 

\item  
We propose a truthfulness notion called {\em lex-truthfulness} for randomized ordinal  
mechanisms. This is stronger than a notion studied in the literature, and yet 
flexible enough that it admits a rich class of mechanisms 
{\em bypassing classical impossibility results}.  
We provide a characterization result for lex-truthfulness, which we leverage to obtain 
lex-truthful mechanisms for various ordinal settings.   
We believe that this characterization will find application beyond the specific
applications that we consider in this paper.

\item 
We demonstrate the usefulness of the above two notions by devising lex-truthful mechanisms
achieving good rank-approximation factors both in the general ordinal setting, as well as 
structured settings such as {\em (one-sided) matching markets}, and its generalizations,
{\em matroid} and {\em scheduling} markets.
\end{list}
We now elaborate on our contributions. 
Let $n$ and $m$ denote the number of agents and number of outcomes respectively,
and $\succ_j$ denote agent $j$'s ordering over outcomes, which we assume is strict
and complete (i.e., for any two outcomes $o, o'$, either $o\succ_j o'$ or vice versa).

\paragraph{Rank approximation (\mySec{rankapx})} 
We say that an outcome $o$ has {\em rank approximation} $\al$ for preference profile
$\bsucc$, if  for {\em every} position $r$, 
the number of agents having $o$ as one of their top-$r$ outcomes 
is at least $\frac{1}{\al}\cdot\mrank_r(\bsucc)$, where $\mrank_r(\bsucc)$ denotes 
$\max_{\ho}$(number of agents having $\ho$ as one of their top-$r$ outcomes). 
An {\em $\al$-rank-approximation mechanism} is one that always returns an
$\al$-rank-approximate outcome. 
While the requirement
of simultaneously approximating $\mrank_r(\bsucc)$ for all $r$ seems too
stringent, and even the {\em existence} of an $\al$-rank-approximate outcome $o$, for
non-trivial $\al$, 
seems doubtful,  
promisingly (as we elaborate later), we can achieve a 
$2$-rank-approximation for matching and matroid markets, and a randomized 
$O(\log n)$-rank-approximation for general ordinal settings.

Rank approximation is a natural, purely ordinal notion with various desirable properties. 
Consider any cardinal-utility profile $\vec{U}=(U_1,\ldots,U_n)$, where each $U_j$ is
consistent with $\succ_j$, that is, $U_j(o)>U_j(o')$ iff $o\succ_j o'$. Call such a
utility profile {\em homogeneous}, if for all $r=1,\ldots,m$, all $U_j$s assign the same
value to their $r$-th ranked outcome. 
An $\al$-rank-approximation outcome $o$ for $\bsucc$ is such that for 
{\em any} consistent homogeneous utility profile $\vec{U}$, 
its social welfare, $\sum_{j=1}^n U_j(o)$, for $\vec{U}$ 
is at least a $\frac{1}{\al}$-fraction of the optimum social welfare 
for $\vec{U}$.  
Thus, an $\al$-rank-approximation mechanism {\em simultaneously yields an
$\al$-approximation to the optimum social-welfare for all consistent homogeneous utility   
profiles} (Theorem~\ref{homutil}). 

Consistent homogeneous utilities are also known as {\em scoring rules}~\cite{Young75}
(also sometimes called positional scoring roles). 
A scoring rule assigns a score to each position and returns the outcome with highest total 
score; 
a prominent example is the {\em Borda rule}, 
which gives a score of $m-k$ to the $k$-th position. 
An outcome is $\alpha$-approximate with respect to a scoring rule, if its score is at
least a $\frac{1}{\alpha}$-fraction of the score of any other outcome. 
Translated to this setting, we obtain 
that an $\al$-rank-approximation mechanism 
{\em simultaneously achieves an $\al$-approximation to all scoring rules}. 
In other words, given an $\al$-rank-approximation mechanism $\cM$, one need not be overly
concerned about which scoring rule is most suited to the problem, since $\cM$ guarantees
an $\al$-approximation to all scoring rules!

To place these simultaneous-approximation bounds 
in perspective, it is useful to consider
an even stronger notion: 
say that a mechanism has ``strong
welfare factor'' $\alpha$, if for every consistent (even non-homogeneous) cardinal-utility
profile $\vec{U}$, the mechanism returns an $\al$-approximation to the optimum social
welfare for $\vec{U}$. 
Not surprisingly, 
this notion is too strong: it is easy to show that no mechanism (deterministic or
randomized) can have any non-trivial strong welfare factor, even for matching markets. 

\paragraph{Lex-truthfulness (\mySec{lextruth})} 
The classic impossibility results of~\cite{Gibbard73,Satterthwaite75}
show that the space of deterministic truthful mechanisms in general ordinal settings is
extremely limited, forcing the move to randomized mechanisms.
When seeking to define a notion of truthfulness for ordinal randomized mechanisms, one
immediately encounters the following issue: {\em how should one extend an agent's
preferences over outcomes to preferences over distributions of outcomes?} 
The usual approach in the economics literature is to use the 
{\em stochastic dominance} relation. Since this 
does not induce a total order over distributions, one obtains two notions of
truthfulness: (i) {\em strong truthfulness}~\cite{Gibbard78}, where the truth-telling
distribution stochastically dominates any distribution obtained via a misreport; and 
(ii) {\em weak truthfulness}~\cite{PS86,BM01}, where the truth-telling distribution is not
stochastically dominated by any distribution obtained via a misreport.
Gibbard~\cite{Gibbard78} generalized~\cite{Gibbard73,Satterthwaite75} 
to show that the space of strongly-truthful mechanisms 
in general ordinal settings is also limited, leaving weak-truthfulness as the only viable
notion of truthfulness for randomized mechanisms. 

We propose a new notion of truthfulness sandwiched (strictly) between the above two
notions. A distribution $\bp$ lex-dominates a distribution $\bq$ with
respect to ordering $\succ$, if, when considering outcomes in decreasing order of their
ranking in $\succ$,
at the first outcome $o$ where $\bp$ and $\bq$ differ, $\bp$ assigns a
higher probability to $o$ than $\bq$. 
Note that lex-dominance induces a {\em total order} on distributions.  
We say that a mechanism is {\em lex-truthful} (LT) if
no distribution obtained by a misreport lex-dominates the truth-telling distribution.%
\footnote{We have recently learned that this notion was independently proposed by
  Cho~\cite{Cho12}, who called it DL-strategyproofness.}  

We show that lex-truthfulness provides us with ample flexibility in mechanism design  
and allows us to {\em circumvent Gibbard's impossibility theorem}. 
Call a social choice function (SCF) $f$ 
{\em fully lex-truthfully (LT) implementable} if for all $\eps>0$, there exists a
lex-truthful mechanism that agrees with $f$ with probability at least $(1-\eps)$ on every
preference profile. 
We isolate a property of an SCF, that we call {\em pseudomonotonicity}, that 
{\em completely characterizes} LT-implementability of the SCF (Theorem~\ref{pseudo}).  
Roughly speaking, an SCF is pseudomonotone if for any preference profile, if an agent
$j$ changes his ordering without altering his top $k$ choices, then the new outcome cannot 
both be a better outcome for $j$ and a top-$(k+1)$ outcome for $j$ (see
Definition~\ref{pseudodefn}).  

This characterization turns out to be instrumental in making lex-truthfulness an amenable
notion to work with, 
and opens up a host of SCFs to full LT-implementation.
We show that various rank-approximation SCFs that we devise for matching, matroid, and
scheduling markets---including the 2-rank-approximation mechanism for matching markets
mentioned earlier---are pseudomonotone. 
For general ordinal settings, we identify a rich class of pseudomonotone SCFs which
includes the {\em plurality scoring rule}. Thus, {\em all} of these SCFs are fully
LT-implementable. We view the characterization of lex-truthfulness via pseudomonotonicity
as one of our main contributions, which we believe will find further applications. 

\paragraph*{Matching, matroid, and scheduling markets (\twoSecs{matching}{sched})} 
In addition to general ordinal mechanism-design settings, we also consider various
structured settings, and obtain lex-truthful mechanisms with good
rank-approximation factors.

Our most-compelling results are for {\em matching markets} (\mySec{matching}), which
are one of the most well-studied ordinal settings (see, e.g., the
surveys~\cite{SU08,AS10}). 
Here, we have $n$ agents and $m$ items, and outcomes are matchings of agents to
items. Each agent has a strict preference over items, which induces preferences
over matchings based on the item the agent is assigned in a matching. 
Observe that agents are indifferent over outcomes that give them the same item. 
The room allocation problem is an instance of this market.  

We devise a simple deterministic 2-rank-approximation
pseudomonotone algorithm \mma (Theorem~\ref{matching-ub}), which is therefore fully
LT-implementable. In contrast, we show in \myApp{known-mat} that various common
algorithms proposed for matching markets, such as the top-trading-cycle algorithm,
randomized serial dictatorship, probabilistic serial,
{\em all have rank approximation at least $\Omega(\sqrt{n})$}. We 
prove a matching lower bound of 2 on the rank-approximation factor of deterministic 
SCFs (Theorem~\ref{matching-lb}), and obtain super-constant lower bounds on the
rank-approximation factor achievable by deterministic truthful mechanisms.

The 2-rank-approximation for matching markets extends to {\em matroid markets}
(Theorem~\ref{matroid-thm}), which is the generalization where we have a matroid on the
agent-set for every item, and the (possibly multiple) agents assigned to an item
are required to form an independent set in that item's matroid. 
Besides the increased modeling power of matroids, this turns out to be a key
\nolinebreak
component of our algorithms for scheduling markets.

In \mySec{sched}, we consider scheduling markets. Here the agents are jobs that need
to be assigned to machines. Each job has a {\em private} ordering over the machines, and a
public processing time on each machine, and there is makespan bound $T$ that limits the
amount of time available 
on each machine. An outcome is a partial assignment of some jobs to machines satisfying
the makespan bound. 
This can be viewed as the matching problem with a {\em knapsack constraint}. 
For parallel machines, we obtain an LT-mechanism that always returns an 
$O(\log n)$-rank-approximation schedule with $O(T)$ makespan, and we show that this bound
is tight (Theorems~\ref{parallel-ub2} and~\ref{parallel-lb}). 
We also obtain an $O(\log n)$-rank approximation for unrelated machines
(Theorem~\ref{unrelated}), albeit not via an LT mechanism.

\subsection{Other related work}\label{sec:related} \label{related}
The conundrum of social welfare in ordinal mechanisms, which probably has its origins in the 
Condorcet paradox~\cite{Condorcet1785} that states that 
it may so happen that a majority of agents prefer outcome $a$ to $b$, outcome $b$ to $c$, and outcome $c$ to
$a$, was cemented by Arrow's impossibility theorem~\cite{Arrow51}.  
Subsequent to Arrow's result, much of the work in social choice theory has focused
on Pareto optimality as the sole notion of efficiency for ordinal mechanisms. 

Recent work, mostly in the CS literature, has led to a more nuanced notion of efficiency.
Procaccia and Rosenchein~\cite{PR06} studied the strong welfare factor notion (that they call
distortion), and noticed that deterministic mechanisms have unbounded distortion.   
Boutilier et al.~\cite{BCH+12} proposed randomized mechanisms and showed that the strong
welfare factor is at most $O(\sqrt{m}\log^*m)$, if the consistent cardinal-utility profile
is normalized. 
In contrast, our rank approximation results imply $O(\log n)$-approximate outcomes, but
under a stronger restriction on the consistent cardinal utilities.
The notion of approximations to scoring rules was studied by
Procaccia~\cite{Procaccia10} where he described strongly truthful mechanisms 
which $2$-approximate Borda, but $O(\sqrt{m})$-approximate the plurality rule.
In contrast, our (non-truthful) mechanism $O(\log n)$-rank approximates {\em any} scoring
rule, and plurality can be arbitrarily well approximated by a lex-truthful mechanism.

Another notion of social welfare in ordinal mechanisms, called ordinal welfare factor
(OWF), was recently proposed by Bhalgat et al~\cite{BCK11}. 
A mechanism has OWF $\beta\in [0,1]$ if for any outcome $o$, at least $\beta n$
agents prefer the outcome returned by the mechanism to $o$. This is in
fact a {\em quantification} of the notion of  {\em popular} outcomes; an outcome is
popular if a majority prefer it to any other fixed outcome. Note that popular outcomes
have OWF of at least $0.5$. 
A popular outcome may not exist, but a popular distribution over outcomes always does.
Popular outcomes were studied by economists in the matching setting~\cite{Gardenfors75},
and as {\em  strict maximal lotteries} in the general
setting~\cite{Fishburn84,Kreweras65}; subsequently, a 
large body of literature has been developed by computer scientists on popular
 matchings~\cite{AIK+07, KMN10, HKM+08, Mahdian06}. 
The notions of rank approximation and OWF (and therefore the notion of
popularity) are incomparable. That is, there are outcomes with ``good'' OWF and ``bad'' 
rank approximation, and vice-versa. 

Subsequent to the Gibbard-Satterthwaite result, researchers focused on design of
randomized mechanisms. As mentioned above, this led
to differing notions of truthfulness. Strong truthfulness was proposed by
Gibbard~\cite{Gibbard78}. Postlewaite and
Schmeidler~\cite{PS86} proposed 
weak truthfulness and proved that no weakly
truthful mechanism on $4$ or more outcomes, can be (ex ante) Pareto
optimal if agents are allowed to have priors on
their (own) preferences. 
Subsequently, Aziz et al~\cite{ABB13} removed the prior condition, but prove
impossibility of only certain {\em kinds} of mechanism. We remark that our
lex-truthful mechanisms, which are also weakly truthful, do not contradict 
these results, since our mechanisms are not Pareto optimal.
However, our mechanisms are $\eps$-implementations of Pareto-optimal SCFs,
so they satisfy Pareto optimality with probability at least $1-\eps$. Thus, we bypass
the above impossibility results while sacrificing a modicum of Pareto-optimality.

Matching markets are one of the most widely studied examples of the ordinal setting. There
is a vast amount of literature, and we point to excellent surveys~\cite{RS90, SU08, AS10}. 
In \myApp{known-mat}, we describe three well known mechanisms in this setting. These are the
random serial dictatorship, Gale's top trading cycle algorithm~\cite{SS74}, and the
probabilistic serial (PS) mechanism~\cite{BM01}. The first two mechanisms are at least
strongly truthful. PS is weakly truthful, and we show that it is lex-truthful as well;
this was also independently shown by~\cite{Cho12,SchulmanV12}.
However, we show that all these three mechanisms have rank approximation as bad as 
$\Omega(\sqrt{n})$. In contrast, we obtain a fully LT-implementable 2-rank-approximation
mechanism using our pseudomonotone 2-rank-approximation algorithm \mma.

\section{Preliminaries} \label{sec:prelim} \label{prelim}
In the general {\em ordinal mechanism design} setting, we have a set $N$ of $n$ agents, 
and a set $O$ of $m$ outcomes (or alternatives). We use the terms agent and player
interchangeably. Each agent $j\in N$ has a {\em private} complete preference list or
ordering $\succeq_j$ over outcomes, that is, $o\succeq_j o'$ or $o'\succeq_j o$ for every  
$o,o'\in O$. This is typically referred to as {\em ordinal} utilities/preferences, to
distinguish them from {\em cardinal} utilities wherein the utility function assigns a
value to each outcome.  
Let $\Sg_j$ denote the publicly-known set of allowed preference lists for agent $j$, and 
$\Sg:=\prod_{j=1}^n\Sg_j$. 
A preference profile is a combination $\bsucceq=(\succeq_1,\ldots,\succeq_n)$ of agents'
preference lists. 
For $k\in\Z_+$, we use $[k]$ to denote the set $\{1,\ldots,k\}$.
A preference list is called {\em strict}, and denoted $\succ$, if there are no 
indifferences: 
exactly one of $o\succ o'$ and $o'\succ o$ holds for every two distinct outcomes 
$o,o'\in O$. 
Given a strict preference $\succ$, we will sometimes say $o\succeq o'$ to denote that
$o\succ o'$ or $o=o'$.
Given a preference list $\succ$, let $\alt{\succ}{r}\in O$ denote the $r$-th ranked outcome in
$\succ$, and $\pos{\succ}{o}\in[m]$ denote the rank of outcome $o$ in $\succ$.
For a tuple $x=(x_1,\ldots,x_n)$, we use $x_{-j}$ to denote
$(x_1,\ldots,x_{j-1},x_{j+1},\ldots,x_n)$. Similarly, let $\Sg_{-j}:=\prod_{k\neq j}\Sg_k$.

In addition to the general setting mentioned above, we consider three specific
mechanism-design settings in this paper: one-sided {\em matching markets}, 
which have been studied extensively in the literature (see, e.g.,~\cite{SU08, AS10}) and two
generalizations of these, {\em matroid markets} and {\em scheduling markets}, that we
introduce. 

\paragraph{Matching markets (\mySec{matching})}
We nave $n$ agents and $m$ items. Each agent $j$ has a strict preference $\succ_j$ over the
$m$ items. The outcomes are matchings of agents to items.
We say that an outcome $M$ assigns an agent $j$ the ``null'' item $\es$ to denote that he
is not assigned an item in $M$; we set $i\succ_j\es$ for every item $i$.
An agent is indifferent between matchings $M$ and $M'$ if they allot him the same item
(counting $\es$ as an item), and otherwise, prefers $M$ to $M'$ if he prefers the
item allotted to him in $M$ to the item allotted to him in $M'$.  

\paragraph{Matroid markets (\mySec{matroid})}
We again have $n$ agents who have a strict preference over $m$ items. 
We also have a matroid $\mat_i=(N,\cI_i)$ on the set $N$ of agents, for each item $i\in[m]$. 
An outcome is an allocation that assigns at most one item to each agent $j$ such that, for
each item $i$, the set of agents allotted item $i$ is an independent set of $\mat_i$. 
Note that multiple agents may be allocated the same item. 
An agent's ordering over outcomes is induced by his ordering of the items as in the
setting of matching-markets.
It is easy to see that a matching market is the special case where $\mat_i$ encodes that
at most one agent may be assigned to item $i$. 

\paragraph{Scheduling markets (\mySec{sched})}
The agents are $n$ jobs that need to be scheduled on $m$ machines, where
the machines are in general {\em unrelated}. Each job $j$ possesses a {\em private}
strict complete preference order $\succ_j$ over the machines, and has a 
{\em publicly-known} processing time $p_{ij}$ on machine $i$. Furthermore, there is a
bound $T$ on the maximum load allowed on any machine (i.e., makespan).  
An outcome is an (partial) assignment of some jobs to machines that respects the makespan 
bound. 
The ordering over outcomes is induced by the ordering over machines as in the above two
cases. The parallel machines setting is the case where $p_{ij} = p_j$ for every machine
$i$ and job $j$.

\smallskip
Note that in the above three markets, agents' preferences over outcomes are 
{\em not} strict; however, for each agent $j$, the outcome-set may be partitioned into
{\em indifference classes} such that $j$ is indifferent between the outcomes in an
indifference class, and has a strict ordering over the indifference classes. 
Our framework 
and results apply to such settings with cosmetic notational changes 
(see \mySec{indiff}),
but we stick for the most part to the setting
of strict preferences for notational ease. 

A {\em social choice function} (SCF) is a function $f:\Sg\mapsto O$. 
In settings with no monetary transfers, there is no formal distinction between
an SCF and a {\em deterministic} algorithm or {\em direct-revelation} mechanism, which
maps the preference profile given by the agents' reported preference lists to an outcome. 
An SCF $f$ is said to be {\em implementable} or {\em truthful} 
if for every player $j$, every $\succ_j, \succ'_j\in\Sg_j$, and every
$\bsucc_{-j}\in\Sg_{-j}$, 
we have  $f(\succ_j,\bsucc_{-j})\succeq_j f(\succ'_j,\bsucc_{-j})$; that is, no agent
benefits by misreporting his preference list.

A {\em randomized mechanism}  
is allowed to output a distribution (also called a {\em lottery})
over outcomes. Let $\cL(O)$ denote the collection of distributions over the outcome-set
$O$. A randomized mechanism is formally then a function 
mapping preference profiles to distributions in $\cL(O)$. 
We sometimes refer to a mechanism 
that works with ordinal preferences as an ordinal mechanism.

\begin{definition}
A randomized mechanism $\cM$ is said to {\em $\eps$-implement} an SCF $f$ (or that $f$ is
$\ve$-implementable by $\cM$), if $\Pr[\cM(\bsucc)=f(\bsucc)]\geq 1-\eps$ for all
$\bsucc\in\Sg$, where the probability is over the random choices of $\cM$.
We say that a family $\{\cM^\ve\}$ of mechanisms {\em fully implements} $f$ if for all
$\ve>0$, $\cM^\ve$ $\ve$-implements $f$.
(This is in the same spirit as the notion of virtual implementation in Nash
equilibrium~\cite{Matsushima88,AbreuS91}.)    
\end{definition}

Truthfulness for randomized mechanisms may be defined in various ways. The strongest
notion is {\em universal truthfulness}, wherein
a randomized truthful mechanism is a randomization (or mixture) over deterministic
truthful mechanisms, where the mixture weights are input-independent. 
A somewhat weaker notion is obtained by considering the stochastic dominance relation.
Given an ordering $\succ$ over $O$, and two lotteries $\bp,\bq\in \cL(O)$, we say that
$\bp$ (first-order) {\em stochastically dominates} $\bq$ with respect to $\succ$, if
$\sum_{\ell\leq i}\bp(\alt{\succ}{\ell}) \geq \sum_{\ell\leq i} \bq(\alt{\succ}{\ell})$
for all $i=1,\ldots,m$. 
Since stochastic dominance does not induce a total ordering on $\cL(O)$, this yields two
notions of truthfulness that have been studied in the literature.  

\begin{definition} 
A randomized mechanism $\cM$ is said to be:
\begin{list}{$\bullet$}{\itemsep=0ex \topsep=0.5ex \addtolength{\leftmargin}{0ex}} 
\item {\em strongly truthful}~\cite{Gibbard78}: if 
$\cM(\succ_j,\bsucc_{-j})$ {\em stochastically dominates} $\cM(\succ'_j,\bsucc_{-j})$ with
respect to $\succ_j$,  for all $j$, all $\succ_j,\succ'_j\in\Sg_j$, and all
$\succ_{-j}\in\Sg_{-j}$. 

\item {\em weakly truthful}~\cite{PS86,BM01}: if 
$\cM(\succ_j,\succ_{-j})$ {\em is not stochastically dominated} by
$\cM(\succ'_j,\succ_{-j})$ with respect to $\succ_j$, for all $j$, all
  $\succ_j,\succ'_j\in\Sg_j$, and all $\succ_{-j}\in\Sg_{-j}$. 
\end{list}
\end{definition}

A universally truthful mechanism is also strongly truthful, and in fact, this
inclusion is strict (Theorem~\ref{hier}). 
Gibbard~\cite{Gibbard78} extended the impossibility result of~\cite{Gibbard73,Satterthwaite75}
to show that the space of strongly truthful mechanisms is also rather limited. 
A deterministic mechanism is: 
(i) {\em dictatorial} if there exists $j\in N$ such that the mechanism's output is always
$j$'s top choice; and (ii) {\em duple} if the mechanism's range $f(\Sg)$ consists of at
most two outcomes.
A (deterministic or randomized) mechanism is {\em unilateral} if there exists some fixed
$j\in N$ such that the mechanism's output depends only on $j$'s (reported) preference list.

\begin{theorem}(Gibbard-Satterthwaite and Gibbard impossibility
    results)\label{thm:imposs} \label{imposs}
(i) If $m\geq 3$ and $f(\Sg)=O$, then $f$ is truthful iff it is dictatorial.
(ii) Any strongly truthful mechanism is a mixture of truthful unilateral and deterministic
truthful duple mechanisms with input-independent mixture weights.
\end{theorem}

Theorem~\ref{imposs} leaves weak truthfulness as the only notion that
potentially allows for some sophisticated mechanisms. In
\mySec{lextruth}, we propose a stronger notion of truthfulness 
and show that this is flexible enough that one can 
{\em bypass Gibbard's impossibility result} and obtain various interesting mechanisms
including, in particular, mechanisms that yield ``good'' social welfare under the metric
we introduce in \mySec{rankapx}.

\section{Rank approximation and Lex-truthfulness} \label{ranklex}

\subsection{Welfare in ordinal settings: rank approximation} \label{rankapx} 
We introduce a notion of social welfare that we call {\em rank approximation}.
Given a preference profile $\bsucc=(\succ_1,\ldots,\succ_n)$, the $i$-rank of an outcome
$o\in O$ in $\bsucc$, denoted $\rank_i(o;\bsucc)$, is the number of agents having $o$ in
their top $i$ choices:
$\rank_i(o;\bsucc) = \bigl|\{j: \pos{\succ_j}{o} \leq i\}\bigr|$.
Define {$\maxrank_i(\bsucc):=\max_{o\in O}\rank_i(o;\bsucc)$.}

\begin{definition} \label{rankapxdef} 
A randomized mechanism $\cM$ is an {\em $\al$-rank-approximation} mechanism, if for
every preference profile $\bsucc$, we have
$\Exp\bigl[\rank_i(\cM(\bsucc);\bsucc)\bigr]\geq\mrank_i(\bsucc)/\al$ for {\em all}
$i=1,\ldots,m$, where the expectation is taken over the random choices of $\cM$.  
We say that $\al$ is the rank-approximation factor of $\cM$. 
\end{definition}

As mentioned in the Introduction, rank approximation is an appealingly robust notion from
various perspectives.
A utility function $U$ is consistent with a preference ordering $\succ$ if $U(o) > U(o')$ whenever $o\succ o'$.
A collection of cardinal utility functions $(U_1,\ldots,U_n)$ consistent with a
preference profile $\bsucc$ is called {\em homogeneous} if for all $i\in[m]$, the value
that an agent assigns to his $i$-th choice is the same across all agents, that is, 
$U_j(\alt{\succ_j}{i})=U_{j'}(\alt{\succ_{j'}}{i})$ for all $i\in[m], j,j'\in N$. 

An $\alpha$-rank-approximation mechanism yields an $\alpha$ approximation to social
welfare for {\em any} homogeneous cardinal-utility profile consistent with the
agents' preferences.  

\begin{theorem} \label{thm:homutil} \label{homutil}
Let $\cM$ be an $\al$-rank-approximation randomized mechanism. Then, for every preference  
profile $\bsucc$, we have $\Exp\bigl[\sum_{j\in N} U_j\bigl(\cM(\bsucc)\bigr)\bigr]\geq
\frac{1}{\al}\cdot\max_{o\in O}\sum_{j\in N} U_j(o)$ 
for any homogeneous utility profile $(U_1,\ldots,U_n)$ consistent with $\bsucc$. 
\end{theorem}

\begin{proof}
Let $\bp=\cM(\bsucc)$. Let $U(i)$ be the common value of $U_j(\alt{\succ_j}{i})$. 
Define $\rank_0(o;\bsucc)=0$ for all $o\in O$, and $U(m+1)=0$.
Let $o^*=\argmax_{o\in O}\sum_{j\in N}U_j(o)$.
Then $\Exp\bigl[\sum_{j\in N} U_j\bigl(\cM(\bsucc)\bigr)\bigr]$ is
\begin{equation*}
\begin{split}
& \sum_{o\in O}\bp(o)\sum_{i=1}^m\bigl(\rank_i(o;\bsucc)-\rank_{i-1}(o;\bsucc)\bigr)U(i)
=\sum_{o\in O}\bp(o)\sum_{i=1}^m\rank_i(o;\bsucc)\bigl(U(i)-U(i+1)\bigr) \\
& = \sum_{i=1}^m\bigl(U(i)-U(i+1)\bigr)\Exp\bigl[\rank_i(\cM(\bsucc);\bsucc)\bigr]
\geq \frac{1}{\al}\cdot\sum_{i=1}^m\bigl(U(i)-U(i+1)\bigr)\rank_i(o^*;\bsucc)
=\frac{1}{\al}\cdot\sum_{j\in N}U_j(o^*). \qedhere 
\end{split}
\end{equation*} 
\end{proof}

\smallskip
Consistent homogeneous utilities may be equivalently viewed as a 
{\em scoring rule}; 
Viewed from this perspective, 
Theorem~\ref{homutil} shows that 
an $\al$-rank-approximation mechanism 
{\em simultaneously achieves an $\al$-approximation to all 
 scoring rules}. 

In fact, rank approximation satisfies an even more general robustness property.
Associate with each outcome $o$ an $m$-vector called its {\em histogram}, given by 
$\hist(o;\bsucc)=\{\rank_i(o;\bsucc)\}_{i\in[m]}$. 
Then the rank-approximation factor of an outcome $o$ is
$g\bigl(\hist(o;\bsucc);\bsucc\bigr)$, 
where $g(x;\bsucc):=\min_{i\in[m]}\frac{x_i}{\mrank_i(\bsucc)}$. 
It is not hard to see that $g$ is a concave function of $x$ and non-decreasing in each
coordinate. A deterministic $\al$-rank-approximation mechanism outputs an
outcome $o$ whose $g$-value, $g\bigl(\hist(o;\bsucc);\bsucc\bigr)$, is 
at least $\frac{1}{\al}$ for every input $\bsucc$. 

Now suppose $h(x;\bsucc)$ is {\em any} concave non-decreasing function and we measure 
the value of an outcome $o$ by $h\bigl(\hist(o;\bsucc);\bsucc\bigr)$. This yields a
natural SCF $f^h$, where $f^h(\bsucc)$ is \linebreak 
$\argmax_{o'\in O}h\bigl(\hist(o';\bsucc);\bsucc\bigr)$. 
Note that scoring rules correspond to the special case where $h(\cdot)$ is linear with all
coefficients non-negative. 
Analogous to $\al$-rank-approximation, we can define an SCF $f'$ to be an
$\al$-approximation for $f^h$ if 
$(\text{$h$-value of $f'(\bsucc)$})\geq
\frac{1}{\al}\cdot(\text{$h$-value of $f^h(\bsucc)$})$ for all $\bsucc$.

An deterministic $\al$-rank-approximation mechanism {\em simultaneously} achieves an
$\al$-approximation mechanism for {\em all} such histogram-based concave SCFs: 
if $o$ is the outcome returned, we 
have
$\hist(o;\bsucc)\geq\frac{1}{\al}\cdot\hist(o';\bsucc)$ coordinate-wise for any $o'\in O$.
Since $h$ is non-decreasing and concave, this implies that
$h\bigl(\hist(o;\bsucc);\bsucc\bigr)\geq
\frac{1}{\al}\cdot \max_{o'\in O} h\bigl(\hist(o';\bsucc);\bsucc\bigr) = \frac{1}{\al}\cdot f^h(\bsucc)$.

\subsection{Truthfulness for randomized ordinal mechanisms: lex-truthfulness} \label{lextruth} 
We propose a new notion for truthfulness relying on lexicographic ordering.
Given an ordering $\succ$ over $O$, and two lotteries $\bp\neq \bq\in \cL(O)$, $\bp$ 
{\em lexicographically dominates} $\bq$ with respect to $\succ$,  
if there exists $i\in[m]$ such that $\bp(\alt{\succ}{i})>\bq(\alt{\succ}{i})$ and 
$\bp(\alt{\succ}{\ell})=\bp(\alt{\succ}{\ell})$ for all $\ell=1,\ldots,i-1$. 
Note that 
lex-dominance imposes a 
{\em total order} on $\cL(O)$. This motivates the following definition of truthfulness.   

\begin{definition} \label{lext}
A randomized mechanism $\cM$ is called {\em lex-truthful} (LT) if for all $j\in N$, all 
$\succ_j,\succ'_j\in\Sg_j$, and all $\bsucc_{-j}$, we have that either
$\cM(\succ_j,\succ_{-j})=\cM(\succ'_j,\succ_{-j})$, or  
$\cM(\succ_j,\succ_{-j})$ {\em lexicographically dominates}
$\cM(\succ'_j,\succ_{-j})$ with respect to $\succ_j$. 
\end{definition}

Observe that if $\bp$ stochastically dominates $\bq$, then
$\bp$ lex-dominates $\bq$ as well. 
Since lex-dominance is a total order, this implies that if $\bp$ lex-dominates $\bq$, then
$\bq$ cannot stochastically dominate $\bp$. 
We obtain the following hierarchy between the various notions of truthfulness for
randomized ordinal mechanisms.

\begin{theorem}\label{thm:hier} \label{hier}
Let $\univt,\ \strongt,\ \weakt,\ \lext$ denote the classes of universally-, strongly-, 
weakly-, and lex- truthful mechanisms respectively.
Then $\univt\subsetneq\strongt\subsetneq\lext\subsetneq\weakt$.
\end{theorem}

We defer the proof of Theorem~\ref{hier} to \myApp{append-ranklex}.
We shorten ``implementable by a lex-truthful mechanism'' to 
``lex-truthfully (LT) implementable'' in the sequel.
We show that lex-truthful implementability is equivalent to a property of the 
social-choice function that we call {\em pseudomonotonicity}. This characterization
immediately opens up a host of SCFs that are fully LT-implementable. 
We heavily exploit this in \twoSecs{matching}{sched} to show that the
rank-approximation SCFs that we devise for various problems are fully
LT-implementable. In \mySec{general}, we leverage this to show that an interesting 
class of SCFs in general ordinal settings are fully LT-implementable.

\begin{definition} \label{pseudodefn}\label{def:pseudodefn} 
A social choice function $f$ is {\em pseudomonotone} (or satisfies 
{\em pseudomonotonicity}) if the following holds. Consider any player $j$, 
$\bsucc_{-j}\in\Sg_{-j}$, and $\succ_j, \succ'_j\in\Sg_j$. 
Let $o = f(\bsucc)$ and $o' = f(\bsucc')$. Then, 
either (i)  $o \succeq_j o'$, or 
(ii) there is an outcome $o''$ such that $o''\succ_j o'$ and $\pos{\succ_j}{o''}<\pos{\succ'_j}{o''}$.
\end{definition}

A useful way to view pseudomonotonicity is as follows: if a player's deviation leaves his
first $k$ preferences unaltered, then the deviation cannot both yield him a better outcome 
{\em and} a top-$(k+1)$ outcome.

\begin{theorem}\label{thm:pmtolt}\label{pseudo}\label{thm:pseudo}
(i) Let $f$ be a pseudomonotone SCF. Then $f$ is $\eps$-implementable by a lex-truthful
mechanism for any $\eps>0$; that is, $f$ is fully lex-truthfully implementable.

\noindent
(ii) Conversely, if $f$ is $\ve$-LT implementable for some $\ve<\frac{1}{2}$,
then $f$ is pseudomonotone. 
\end{theorem}

\begin{proof}
First consider part (i). 
Given $\eps>0$, one can find $\eps_1 > \eps_2 > \cdots > \eps_m > 0$ such that
$\sum_i\eps_i = \eps$. Consider the randomized mechanism $\cM$ that on input $\bsucc$,
returns $f(\bsucc)$ with probability $(1-\eps)$, and with probability $\eps$ 
it chooses a random agent $a$ and returns his $i$-th preference with probability $\eps_i/\eps$.  

It is clear by definition that $\cM$ $\eps$-implements $f$. To prove lex-truthfulness, fix
an agent $j$ and consider any $\bsucc'=(\succ'_j,\bsucc_{-j})$, where 
$\succ'_j \neq\succ_j$. 
Let $o=f(\bsucc)$ and let $o'=f(\bsucc')$. 
Also let $\bp=\cM(\bsucc),\ \bq=\cM(\bsucc')$. 
Let $\1(A)$ be 1 if $A$ is true, and 0 otherwise.
For any outcome $\ho$, we have 
\begin{equation*}
\bp(\ho) - \bq(\ho) = \frac{1}{n}\left(\eps_{\pos{\succ_j}{\ho}} -
\eps_{\pos{\succ'_j}{\ho}}\right) 
+ \1(\ho=o)\cdot(1-\eps)-\1(\ho=o')\cdot(1-\eps).
\end{equation*} 
Considering outcomes in the preference order of $\succ_j$,
let $o''$ be the first outcome such that $\pos{\succ_j}{o''}\neq\pos{\succ'_j}{o''}$. Then
$\pos{\succ_j}{o''}<\pos{\succ'_j}{o''}$. 
By pseudomonotonicity of $f$, we know that $o\succeq_j o'$ or $o''\succ_j o'$. In the
latter case, we have $\bp(\ho)-\bq(\ho)\geq 0$ for all $\ho\succeq_j o''$ and
$\bp(o'')-\bq(o'')>0$, so we are done. In the former case, if $o=o'$ or $o''\succeq_j o$,
then the same argument holds. So suppose $o\succ_j o''$ and $o\succ_j o'$. 
Then $\bp(\ho)-\bq(\ho)\geq 0$ for all $\ho\succeq_j o$ and $\bp(o)-\bq(o)>0$, so again we
are done. 
(Note that mechanism $\cM$ maps distinct inputs to distinct lotteries, and therefore it
satisfies a slightly stronger version of lex-truthfulness: the truth-telling lottery
lex-dominates (i.e., is {\em strictly} superior to) a lottery obtained via a misreport.) 

\medskip
We now prove part (ii).
Let $\cM$ be an LT mechanism that $\ve$-LT implements $f$, where $\ve<\frac{1}{2}$. 
Suppose for a contradiction, there is some agent $j$, and $\bsucc=(\succ_j,\bsucc_{-j})$ and
$\bsucc'=(\succ'_j,\bsucc_{-j})$ such that $o=f(\bsucc)$ and $o'=f(\bsucc')$ violate the
conditions for pseudomonotonicity. That is, we have $o'\succ_j o$ and for every outcome
$o''\succ_j o'$, we have $\pos{\succ_j}{o''}\geq\pos{\succ'_j}{o''}$. This means that
$\pos{\succ_j}{o''}=\pos{\succ'_j}{o''}$ for all $o''\succ_j o'$.
Let $\bp=\cM(\bsucc)$ and $\bq=\cM(\bsucc')$.

Since $\cM$ $\ve$-LT implements $f$, we have $\bp(o')\leq\ve$ and
$\bq(o')\geq 1-\ve$, so $\bp(o')<\bq(o')$. 
Let $o_1,\ldots,o_r$ be the outcomes $o''$ such that $o''\succ_j o'$ listed in decreasing
preference order according to $\succ_j$. 
Since $\pos{\succ_j}{o''}=\pos{\succ'_j}{o''}$ for all $o''\succ_j o'$, we have
$o_\ell\succ'_j o'$ for all $\ell\in[r]$, and the ordering of the $o_\ell$s is the same in
$\succ_j$ and $\succ'_j$. We claim that $\bp(o_\ell)=\bq(o_\ell)$ for all
$\ell=1,\ldots,r$, which contradicts the fact that $\cM$ is lex-truthful.

We prove the claim by induction on $\ell$. Considering $\succ_j$ to be $j$'s true preference
list, we must have $\bp(o_1)\geq \bq(o_1)$, and considering $\succ'_j$ to be $j$'s true
preference list, we must have $\bq(o_1)\geq\bp(o_1)$. Suppose that $\bp(o_k)=\bq(o_k)$ for
$k=1,\ldots,\ell-1$. Again considering $\succ_j$ and $\succ'_j$ be $j$'s true preference
lists in turn, we obtain that $\bp(o_\ell)=\bq(o_\ell)$.
\end{proof}

\subsection{Settings with indifferences} \label{indiff}
As noted earlier, many of the settings we consider involve 
non-strict preferences. In these settings, the outcome set is partitioned into 
{\em indifference classes} $O^j_1,\ldots,O^j_{m_j}$ for each agent $j$. Agent $j$ is
indifferent between any two outcomes in the same indifference class, and has a strict
complete ordering over his indifference classes that specifies his ordering between two
outcomes in different classes.   
Formally, given $\succeq_j\in\Sg_j$, we define
$\pos{\succeq_j}{o}=r\in[m_j]$ if $o$ lies in the indifference class of $j$ ranked $r$
under $\succeq_j$, and 
$\alt{\succeq_j}{r}\sse O$ is now the indifference class of $j$ ranked $r$ under
$\succeq_j$ (that is, $\{o: \pos{\succeq}{o}=r\}$). 
The preferences induced over outcomes are then: $o\succeq_j o'$ iff
$\pos{\succeq_j}{o}\leq\pos{\succeq_j}{o'}$, and $o\succ_j o'$ iff
$\pos{\succeq_j}{o}<\pos{\succeq_j}{o'}$. 
Say that $o\sim_j o'$ if $o$ and $o'$ belong to the same indifference class of $j$. 

One requires mostly notational changes to extend our framework and results to this
more-general setting. With the above notation in place, the definitions of
$\rank_i(o;\bsucceq)$ (as $|\{j:\pos{\succeq_j}{o}\leq i\}|$, $\mrank_i(\bsucc)$,
rank-approximation (Definition~\ref{rankapxdef}) and pseudomonotonicity
(Definition~\ref{pseudodefn}) remain unchanged.  

We extend lex-dominance and lex-truthfulness as follows. 
Since players have indifference classes 
it is not meaningful to consider probabilities assigned to individual outcomes. 
Instead, we aggregate the probability assigned to an indifference class and define
lex-dominance and lex-truthfulness by considering these aggregate probability vectors. 
Given a lottery $\bp\in\cL(O)$ and $S\sse O$, define $\bp(S):=\sum_{o\in S}\bp(o)$. 
and let $\bp'$ be the aggregated probability vector 
$\bigl(\bp'(j,r):=\bp(O^j_r)\bigr)_{j\in N,r\in[m_j]}$.
Given $\succeq_j\in\Sg_j$, and lotteries $\bp\neq\bq\in\cL(O)$, we say
that $\bp$ lex-dominates $\bq$ with respect to $\succeq_j$ if there exists $r\in[m_j]$
such that $\bp'(j,r)>\bq'(j,r)$ 
and $\bp'(j,\ell)=\bq'(j,\ell)$
for all $\ell=1,\ldots,r-1$. 
Then, as in Definition~\ref{lext}, a mechanism $\cM$ is lex-truthful if for all $\bsucceq$,
all $j\in N$ and all $\succeq'_j\in\Sg_j$, either 
the aggregated probability vectors of $\cM(\bsucc)$ and $\cM(\succeq'_j,\succeq_{-j})$ are
equal, or $\cM(\bsucc)$ lex-dominates $\cM(\succeq'_j,\succeq_{-j})$ with respect to
$\succeq_j$. 

We can now mimic the proof of Theorem~\ref{pseudo} to prove the following analogue for the
above setting, showing that pseudomonotonicity is necessary and sufficient for full
LT-implementability. The proof appears in \myApp{append-ranklex}.

\begin{theorem} \label{newpseudo}
(i) Let $f$ be a pseudomonotone SCF. Then 
$f$ is fully lex-truthfully implementable.

\noindent
(ii) Conversely, if $f$ is $\ve$-lex-truthfully implementable for some $\ve<\frac{1}{2}$,
then $f$ is pseudomonotone. 
\end{theorem}

\section{Matching markets} \label{sec:matching} \label{matching}
Recall that in a matching market there are $n$ agents and $m$ items, and outcomes are
matchings of agents to items. Each agent $j$ has a strict total ordering over items, which
induces his preferences over outcomes: $j$ prefers outcome $o$ to $o'$ if he prefers his
allotted item in $o$ to the one in $o'$.

We show in \myApp{known-mat} that various common mechanisms
proposed in the literature for matching markets 
{all have bad rank approximation}. 
In contrast, we devise a simple deterministic algorithm, \mma, that is a pseudomonotone,    
$2$-rank-approximation algorithm, and hence, yields an LT mechanism
(Theorem~\ref{matching-ub}).   
We complement this by showing two lower bounds. Theorem~\ref{matching-lb}
shows that $2$ is the best rank approximation achievable by {\em any}
deterministic algorithm, proving the tightness of our positive result. 
Next, Theorems~\ref{nobossy} and~\ref{truthlb} demonstrate limitations of 
{\em deterministic} truthful mechanisms for matching markets by showing that such
mechanisms cannot achieve any constant rank approximation.

\vspace{-1.5ex}
\paragraph{Algorithm \boldmath \mma}
Fix a tie-breaking rule over agents. 
On input $\bsucc$, \mma allocates items to agents in $m$ stages.
In stage $r$, we consider the bipartite graph $G_r$ with agents and items as vertices, and
an edge from agent $j$ to item $i$, if $i$ is a top-$r$ item of agent $j$.  
Note that $\maxrank_r(\bsucc)$ is precisely the size of the maximum matching in $G_r$. 
Let $M$ denote the current matching of agents to items (which is $\es$ when $r=1$), which
is a matching in $G_r$. We maintain that at the beginning of stage $r$, $M$ is a 
{\em maximal matching} in $G_{r-1}$; observe this is true when $r=1$.
Since $M$ is a {\em maximal matching} in $G_{r-1}$, an agent has an edge to at most one
item in $G_r\sm M$, where $G_r\sm M$ is the graph obtained from $G_r$ by deleting the
nodes matched by $M$. 
For every unmatched item $i$ that has non-zero degree in $G_r\sm M$ (i.e., $i$ is an
unmatched item that is a top-$r$ item of some unmatched agent)  
we use our tie-breaking rule to pick an agent $j\in G_r\sm M$; we assign item $i$ to $j$
and update $M$. 
Thus, $M$ is updated to a maximal matching in $G_r$. We output the matching at the end of $m$ stages.

\begin{theorem} \label{thm:matching-ub} \label{matching-ub}
\mma a pseudomonotone, $2$-rank approximation algorithm for matching markets, and hence
is fully LT-implementable.
\end{theorem}

\begin{proof}
The 2-rank-approximation guarantee of \mma follows immediately from the fact that \mma
maintains a maximal matching in the ``top-$r$'' graph $G_r$ for all $r$, and the size of
any maximal matching is at least half the size of a maximum matching, and thus at least 
$\mrank_r(\bsucc)/2$.      

Fix an agent $j$. 
Suppose that $j$ deviates from
$\succ_j$ to $\succ'_j$ without altering his top-$r$ items and their ordering, that is, 
$\alt{\succ_j}{\ell}=\alt{\succ'_j}{\ell}$ for all $\ell=1,\ldots,r$, and
$\pos{\succ'_j}{i}>r+1$ for $i=\alt{\succ_j}{r+1}$.
Let $\bsucc=(\succ_j,\bsucc_{-j})$ and $\bsucc'=(\succ'_j,\bsucc_{-j})$.
Since the other agents' inputs have not changed, $\mma(\bsucc)$ and $\mma(\bsucc')$ proceed
identically up to the end of stage $r$. So if $j$ has been assigned an item by this time
(which happens in both runs) we are done. Otherwise, in $\mma(\bsucc')$, all of $j$'s
top-$r$ items are unavailable, and since $j$ demotes $i$ in $\succ'$, edge $(j,i)$ does
not belong to the graph $G_{r+1}$ constructed in stage $r+1$; so $j$ does not obtain $i$
or a top-$r$ item under input $\bsucc'$. This proves pseudomonotonicity.
\end{proof}

\begin{theorem} \label{matching-lb} 
For every $\e>0$, there is a matching market on which every deterministic algorithm has
rank-approximation factor at least $2-\e$.   
\end{theorem}

\begin{proof}
Let $K\geq 1$ be an integer such that $\frac{K}{2K-1}\leq\frac{1}{2-\e}$.
We create an instance with $n=2K-1$ players and items.
We specify the first $K$ preferences of the players; the remaining preferences may
be set arbitrarily. 
Let $\bsucc$ denote the resulting input (with arbitrary remaining preferences). 
Since this is the only input we consider, we drop the $\bsucc$ in $\rank_r(o;\bsucc)$ and 
$\mrank_r(\bsucc)$ in the sequel.
\begin{list}{--}{\itemsep=0ex \topsep=0ex \addtolength{\leftmargin}{-3.5ex}}
\item For $r=2,\ldots,K-1$, the $r$-th preference of a player $j$ is item:
$\left\{
\begin{tabular}{l@{\ ;\qquad}l} 
$r$ & if $j=r$, \\
$K+r-1$ & if $j=r-1$, \\
$r-1$ & otherwise. 
\end{tabular}\right.$
\item The first preference of a player $j$ is: item 1 if $j=1$, and item $n$
otherwise. 
\item The $K$-th preference of a player $j$ is: item $K$ if $j=K-1$, and item
$K-1$ otherwise.
\end{list}

\smallskip
First, we claim that $\mrank_r\geq 2r$ for all $r\in[K-1]$. 
For $r=1$, this is achieved by matching player 1 to item 1, and an arbitrary other player
to item $n$.  
For $r=2,\ldots,K-1$, this is achieved by matching player $r$ to item $r$, each player
$j\in[r-1]$ to item $K+j$, matching one player from $\{r+1,\ldots,n\}$ to item $n$
and $r-1$ other arbitrary players from $\{r+1,\ldots,n\}$ arbitrarily to items in
$[r-1]$. Note that each player is matched to a top-$r$ item in this matching. 
Also, $\mrank_{K}=n$. This is achieved by matching player $K-1$ to
item $K$, each player $j\in[K-2]$ to item $K+j$, matching one player from
$\{K,\ldots,n\}$ to item $n$, and the remaining $K-1$ players from
$\{K,\ldots,n\}$ arbitrarily to items in $[K-1]$.    

Now fix a matching $o$. We show that if $\rank_r(o)>\mrank_r/2$ for
$r=1,\ldots,K-1$, then we must have $\rank_{K}(o)\leq K\leq\mrank_{K}/(2-\e)$. 
Thus, we cannot have $\rank_r(o)>\mrank_r/(2-\e)$ for all $r\in[K]$. 

We show by induction on $r$ that if $\rank_\ell(o)>\mrank_\ell/2$ for all
$\ell\in[r]$, where $r<K$, then $o$ must match player $\ell$ to item $\ell$ for all 
$\ell\in[r]$. 
For the base case, if $\rank_1(o)>\mrank_1/2\geq 1$, then $o$ must match player 1 to item
1 since all other players have item $n$ as their top item. 
For the induction step, suppose that $\rank_\ell(o)>\mrank_\ell/2$ for all
$\ell\in[r]$, where $1<r<K$. Then, by the induction hypothesis, we know that $o$ 
matches player $\ell$ to item $\ell$ for all $\ell\in[r-1]$. We require that
$\rank_r(o)\geq r+1$. Examining the preferences of the players in $\{r,\ldots,n\}$, we
see that for player $r$, items $r$ and $n$ are the only unmatched items in his top-$r$
list, and for a player $j\in\{r+1,\ldots,n\}$, item $n$ is the only unmatched item in
$j$'s top-$r$ list. Therefore, $\rank_r(o)\geq r+1$ is only possible if $o$ matches player
$r$ to item $r$. 

Given the above claim, for players $j=K,\ldots,n$, item $n$ is the only unmatched item in
their top-$K$ list, so $\rank_{K}(o)\leq K$.
\end{proof}

We now show that randomization is necessary to achieve good rank approximation via
truthful mechanisms. 
As a warm up, we first prove a lower bound of $n-1$ on the rank-approximation factor
achievable by 
truthful {\em no-bossy} mechanisms~\cite{Satterthwaite75}.
A {\em no-bossy mechanism} for matching markets is one where no agent can change his
preference and modify the outcome without also modifying his own allocation.

Suppose there are $n$ items. 
Let $\succ^{*} := (1,2,\ldots,n)$ denote the ordering where item $i$ is
the $i$-th ranked item, for all $i\in[n]$.
Let $\succ^{*} \circ (k-1,1)$, denote the preference list that is
identical to $\succ^*$ except that items $(k-1)$ and $1$ are swapped. That is, $(k-1)$ is
the top-item, $1$ is the $k$th-choice, and item $i$ is the $i$-th choice for all 
$i\neq 1, k-1$. 
Given $n$ agents and any set $S\subseteq\{2,3,\ldots,n\}$, let $\bsucc^{S}$ be the
preference profile where each agent $k\in S$ has preference $\succ^{*}\circ (k-1,1)$,
while each $k\notin S$ has preference $\succ^{*}$. 
Thus, $\bsucc^{\es}$ is the preference profile where every agent has the same preference
order $\succ^{*}$ over items. 
For notational convenience, we think of a player who is not assigned an item as being
assigned item $n+1$, which is lower ranked than any (true) item in $[n]$. 

\begin{theorem} \label{thm:nobossy} \label{nobossy}
No deterministic truthful no-bossy mechanism for matching markets can have
rank-approximation smaller than $(n-1)$.
\end{theorem}

\begin{proof}
We consider a matching market with $n$ (agents and) items. 
Let $\cM$ be any deterministic truthful no-bossy mechanism.
Suppose that $\cM(\bsucc^{\es})$ assigns items to $\cN$ players.
By renaming players if necessary, we may assume that $\cM(\bsucc^{\es})$ assigns item $i$
to player $i$ for all $i\in[\cN]$, and item $n+1$ to the remaining players. 

Consider the input $\bsucc^{\{k\}}$.
We claim that $\cM(\bsucc^{\{k\}}) = \cM(\bsucc^{\es})$. 
{\em Due to no-bossiness}, it suffices to show that agent $k$'s allocation is the same in  
$\cM(\bsucc^{\{k\}})$ and $\cM(\bsucc^{\es})$. 
Suppose agent $k$ obtains item $i$ in $\cM(\bsucc^{\{k\}})$. Invoking truthfulness when 
$k$'s true preference list is $\succ^*$ (and the other players' preference lists are
$\succ^*$), we obtain that $k\succeq^* i$, that is, $k\leq i$. Similarly, if $k$'s true
preference list were $\succ^*\circ(k-1,1)$, then truthfulness dictates that 
$i\leq k$. Hence, we have $i=k$.

The above argument can be generalized to show that for any $S\sse[n]$, we have 
$\cM(\bsucc^{S})=\cM(\bsucc^{S\setminus k})=\cM(\bsucc^{\es})$ for all $k\in S$. In 
particular, $\cM(\bsucc^{\{2,\ldots,n\}})$ assigns item $i$ to player $i$ for all
$i\in[\cN]$ and leaves the other players unassigned. 
So under the preference profile $\bsucc^{\{2,\ldots,n\}}$, at most one agent, agent $1$,
gets his top choice; however, assigning every player $j>1$ item $j-1$ yields an outcome
where $n-1$ agents get their top choice.
\end{proof}

While no-bossiness was crucial above, we show via a more sophisticated argument that
no deterministic truthful mechanism can obtain constant rank approximation.   

\begin{theorem} \label{truthlb}
Every deterministic truthful mechanism has rank approximation 
$\Omega\bigl(\frac{\log\log n}{\log\log\log n}\bigr)$. 
\end{theorem}

\begin{proof} 
Let $n$ be large enough so that 
$K:=\floor{\frac{\log\log n}{\log\log\log n}}-2\geq 1$.
We show that on instances with $n$ (agents and) items, no deterministic truthful mechanism
can have rank approximation better than $K$.  

As before, if $\cM(\bsucc^{\es})$ assigns items to $\cN$ players, we may assume that it
matches agent $i$ to item $i$ for $i\in[\cN]$, and the remaining players are unassigned
(i.e., assigned item $n+1$).
Given agents $\{a_1,\ldots,a_k\}$ and integers $r_1,\ldots,r_k\geq 1$, we let
$\bsucc^{(a_1,r_1),(a_2,r_2),\ldots,(a_k,r_k)}$ denote the preference profile where all
agents other than these $a_\ell$'s have preference order $\succ^{*}$, while each $a_\ell$
has preference order $\succ^{*} \circ (r_\ell,1))$. That is, $a_\ell$'s top choice is item
$r_\ell$, his $r_\ell$-th choice is item $1$, and his $i$-th choice is item $i$ for all
$i\neq 1,r_\ell$. 
We show that there exist agents $a_1,\ldots,a_K$ and {\em distinct} integers
$r_1,\ldots,r_K\in[K]$, 
such that, in $\cM(\bsucc^{(a_1,r_1),\ldots,(a_K,r_K)})$, every agent $a_1,\ldots,a_K$ gets
an item whose index is larger than $K$. 
Since all other agents have the same top item, the
number of agents getting their top item is at most $1$.   
This proves that the rank approximation is at least $K$, since assigning item $r_\ell$ to
agent $a_\ell$ for all $\ell\in[K]$, yields an outcome where $K$ agents obtain their
top-choice item.

To find these $K$ agents, we proceed in $K$ stages. In stage $\ell$, we will have a subset  
$S_\ell$ of agents with $|S_\ell|\geq\ell$ having the following property. 
For any $t < \ell$, any $t$ agents $\{a_1,\ldots,a_t\}\subseteq S_\ell$, and for any $t$
distinct integers $r_1,\ldots,r_t\in[K]$, $\cM(\bsucc^{(a_1,r_1),\ldots,(a_t,r_t)})$
allocates all agents in $S_\ell$ an item indexed larger than $K$.  

Note that if we reach stage $K$, then we are done due to the following
reason. Consider any $K$ agents $a_1,\ldots,a_{K}\in S_K$ and any $K$ distinct 
integers $r_1,\ldots,r_{K}\in[K]$. Consider any index $\ell\in[K]$.
Let
$\bsucc'=\bsucc^{(a_1,r_1),\ldots,(a_{\ell-1},r_{\ell-1}),(a_{\ell+1},r_{\ell+1}),(a_{K-1},r_{K-1})}$
and $\bsucc=\bsucc^{(a_1,r_1),\ldots,(a_{K},r_{K})}$.
We know that $\cM(\bsucc')$
allocates all agents in $S_K$ an item indexed larger than $K$. This also implies that 
$o:=\cM(\bsucc)$ allocates $a_\ell$ an item indexed larger
than $K$, otherwise given the preference profile $\bsucc'$,
player $\ell$ has an incentive to deviate from his preference list $\succ^*$ and report
$\succ^*\circ (r_\ell,1)$. Since this holds for all $\ell$, it follows that $o$ allocates
every agent $a_1,\ldots,a_K$ an item indexed larger than $K$.

We now show how to obtain the $S_\ell$ sets. For $\ell<K$, the set $S_\ell$ will satisfy
the stronger property that $|S_\ell|^{\frac{1}{\ell+1}}\geq 2K$ (the reason for this will
become clear later). 
The base case is $S_1 = \{K+1,\ldots,n\}$, which satisfies the stated property.
Given a set $S_\ell$ at the end of stage $\ell<K$ we now show how to construct the set 
$S_{\ell+1} \subseteq S_\ell$. 
We construct the following hypergraph $H_\ell$. The vertices are the agents in
$S_\ell$. The hyperedges are subsets of vertices of size at most $(\ell+1)$ constructed as
follows. For every $\ell$-size subset $\{a_1,\ldots,a_\ell\}$ of $S_\ell$, 
and every $a\in S_\ell$ (which could be the same as one of the $a_t$s), we add the hyperedge
$\{a_1,\ldots,a_\ell,a\}$ if there exist $\ell$ distinct integers
$r_1,\ldots,r_\ell\in[K]$ such
that $\cM(\bsucc^{(a_1,r_1),\ldots,(a_\ell,r_\ell)})$ allocates agent $a$ an
item with index at most $K$.  
Note that the number of hyperedges is at most $|S_\ell|^\ell\cdot K^{\ell+1}$ since 
there are $|S_\ell|$ choices for each $a_t$, and $K$ choices for each $r_t$, and once these
are fixed, there at most $K$ choices for $a$. 

Call a subset $U \subseteq S_\ell$ {\em independent} if no hyperedge is completely contained
in it. Observe that $U$ is a valid input to stage $(\ell+1)$ if $|U|\geq\ell+1$: consider
any $t<\ell+1$ agents $a_1,\ldots,a_t\in U$ and any distinct integers
$r_1,\ldots,r_t\in[K]$. Suppose that $\cM(\bsucc^{(a_1,r_1),\ldots,(a_t,r_t)})$ allocates
some agent $a\in U$ an item with index at most $K$. Then we must have $t=\ell$, otherwise
this would contradict the property assumed of $S_\ell$, and then $\{a_1,\ldots,a_t,a\}$
would be a hyperedge, contradicting independence of $U$.

Lemma~\ref{largeind} shows that there is an independent set $S_{\ell+1}\sse S_\ell$ such
that 
$|S_{\ell+1}|\geq\frac{|S_\ell|^{\frac{1}{\ell+1}}}{K}-1\geq\frac{|S_\ell|^{\frac{1}{\ell+1}}}{2K}$, 
where the last inequality follows since $|S_\ell|^{\frac{1}{\ell+1}}\geq 2K$.
Since $|S_t|^{\frac{1}{t+1}}\geq 2K$ for all $t\leq\ell$, we have 
$$
|S_{\ell+1}|\geq|S_\ell|^{\frac{1}{\ell+1}}/2K
\geq |S_1|^{\frac{1}{(\ell+1)!}}/(2K)^\ell
\geq\Bigl(\frac{n}{2}\Bigr)^{\frac{1}{(\ell+1)!}}/(2K)^\ell
\geq\Bigl(\frac{n}{2}\Bigr)^{\frac{1}{K!}}/(2K)^{K-1}.
$$ 
Moreover, if $\ell+1<K$, then 
$|S_{\ell+1}|^{\frac{1}{\ell+2}}
\geq\bigl(\frac{n}{2}\bigr)^{\frac{1}{(\ell+2)!}}/(2K)^{\ell}
\geq\bigl(\frac{n}{2}\bigr)^{\frac{1}{K!}}/(2K)^{K-2}$.
For $K \leq \frac{\log\log n}{\log\log\log n}-2$, 
we have $\bigl(\frac{n}{2}\bigr)^{\frac{1}{K!}}/(2K)^{K-1}\geq 2K$. 
Hence, $|S_{\ell+1}|\geq 2K$, and if $\ell+1<K$, we have 
$|S_{\ell+1}|^{\frac{1}{\ell+2}}\geq 2K$. 
Thus, we obtain that $|S_K|\geq 2K$. 
\end{proof}

\begin{lemma} \label{largeind}
There exists an independent set $S_{\ell+1} \subseteq S_\ell$ of size 
$|S_{\ell+1}| \geq \frac{|S_\ell|^{\frac{1}{\ell+1}}}{K}-1$. 
\end{lemma}

\begin{proof}
Let $\cN=|S_{\ell}|$. Recall the number of hyperedges is at most $\cN^\ell K^{\ell+1}$.
We first argue that all hyperedges are of size $\ell+1$.
Every hyperedge is of size at least $\ell$.
A size-$\ell$ hyperedge $\{a_1,\ldots,a_\ell\}$ can only arise, if there are $\ell$
distinct integers $r_1,\ldots,r_\ell\in[K]$ and some $a\in\{a_1,\ldots,a_\ell\}$, say
$a_1$ for notational convenience such that
$\cM(\bsucc^{(a_1,r_1),\ldots,(a_\ell,r_\ell)})$ allots $a$ an item indexed less than $K$. 
But the definition of $S_\ell$ implies that 
$\cM(\bsucc^{(a_2,r_2),\ldots,(a_\ell,r_\ell)})$ allots $a_1$ an item with index larger
than $K$. 
This violates truthfulness, since agent $a_1$ has an incentive to misreport 
$\succ^*\circ(r_1,1)$ when his true preference is $\succ^*$ and obtain a better item. 

Consider sampling each vertex of $H_\ell$ with probability $p = K^{-1}\cdot
\cN^{-\left(\frac{\ell}{\ell+1}\right)}$ to get a random subset $X$. If $X$ contains a
hyperedge, then we remove all its vertices from $X$. 
The probability that a hyperedge is present in $X$ is at most $p^{\ell+1}$, since all
hyperedges are of size $\ell+1$. Therefore, in expectation, the size of $X$ after removal
is at least $p\cN-p^{\ell+1}\cN^\ell K^{\ell+1} = \frac{\cN^{1/\ell+1}}{K} - 1$.
\end{proof}

\subsection{A generalization: matroid markets} \label{sec:matroid} \label{matroid}
In this generalization of matching markets, there is a matroid $\mat_i=(N,\cI_i)$ 
on the agent-set $N$ for each item $i$, and multiple agents may be assigned to item $i$
provided they form an independent set of $\mat_i$. 
Here $\cI_i$ is a collection of subsets of $N$ with the
following properties: (i) $\es\in\cI_i$; for all $A, B\sse N$ 
(ii) if $A\in\cI_i$ and $B\sse A$, then $B\in\cI_i$; 
(iii) if $A,B\in\cI_i$ and $|A|>|B|$, then there exists some $j\in A\sm B$ such that
$B\cup\{j\}\in\cI_i$. 
Clearly, the lower bounds obtained for matching markets also
hold in this setting. Complementing this, we extend \mma to obtain a pseudomonotone
2-rank-approximation algorithm for matroid markets. Let $L$ be the set of all items.

\begin{theorem} \label{matroid-thm}
There is a pseudomonotone $2$-rank approximation algorithm for matroid markets, and a
mechanism that fully LT-implements it.
\end{theorem}

\begin{proof} 
The algorithm is similar to \mma. 
Again fix an agent-ordering and an item-ordering. Consider some input $\bsucc$.

We again proceed in $m$ stages. In stage $r$, we consider the ``top-$r$'' graph
$G_r=(N\cup L,E_r)$, where each agent $j$ has edges to his top-$r$ items. 
Note that every outcome induces a feasible solution to the {\em matroid-intersection}
problem defined by the following two matroids on the universe $E_r$.
One is $\mat_A$, which is the direct sum of the $\mat_i$ matroids for all $i\in L$, i.e., 
a set $I\sse E_r$ is independent if $\{j: (j,i)\in I\}\in\cI_i$ for all $i\in L$. 
The second is the partition matroid $\mat_B(r)$ encoding that at most one edge of $E_r$ is
incident to each agent $j$.
Then every outcome induces a set that is independent in both $\mat_A$ and $\mat_B(r)$, and
$\maxrank_r(\bsucc)$ is the size of the largest common independent set.

Let $M$ consist of the edges denoting the current (i.e., at the start of stage $r$)
assignment of items to agents. 
Our algorithm will maintain the invariant that at the end of stage $r$, $M$ is a 
{\em maximal} set that is independent in both $\mat_A$ and $\mat_B(r)$.
The rank-approximation factor of 2 follows then from the well-known fact that every
maximal common independent set of two matroids has size at least half the size of 
maximum-cardinality common independent set; Claim~\ref{matint} gives a self-contained
proof.  

Let $\Gm^r(u)$ denote the neighbors of node $u$ in $G_r$, 
and $\Gm^r_M(u):=\{v: (u,v)\in M\}$.
Note that if $M$ is a maximal common independent set in $G_{r-1}$, then for every agent
$j$ that is not assigned an item in $M$, among $j$'s top-$r$ items his $r$-th ranked item
is the {\em only} item to which $j$ can be possible assigned while preserving independence 
in the item's matroid.

We consider each item $i$ and augment $\Gm^r_M(i)$, the current set of agents assigned
to item $i$, to a maximal subset $J_i\sse\Gm^r(i)$ that is independent in $\mat_i$: we
initialize $J_i$ to $\Gm^r_M(i)$. Next, we consider agents in 
$\Gm^r(i)\sm\Gm^r_M(i)$ according to the fixed agent-ordering and add agent $j$ to $J_i$
if this maintains independence in $\mat_i$. Maximality of $J_i$ follows from the matroid
property. (In fact $J_i$ is a maximum-size independent subset of $\Gm^r(i)$.)
Finally, we update $M$ to reflect the new assignments in stage $r$.

The fact that $M$ is a maximal common independent set of $\mat_A$ and $\mat_B(r)$ is
immediate: if some edge $(j,i)$ can be added to $M$ while preserving independence in
$\mat_A$ and $\mat_B(r)$, then $j$ was unassigned at the start of stage $r$ and when we
considered item $i$, $j$ could (and would) have been added to $J_i$ in the iteration when
$j$ was considered.

We have already argued that the above algorithm is a 2-rank-approximation. 
Pseudomonotonicity of the above algorithm follows from exactly the same arguments as in 
Theorem~\ref{matching-ub}.  
\end{proof}

\begin{claim} \label{matint}
Let $\mat_1(U,\cI_1)$, $\mat_2=(U,\cI_2)$ be two matroids. 
Let $S\sse U$ be an inclusion-wise maximal set that is independent in both $\mat_1$ and
$\mat_2$. Let $A$ be a maximum-cardinality set that is independent in both $\mat_1$ and
$\mat_2$. Then $|S|\geq |A|/2$
\end{claim}

\begin{proof}
Suppose $|S|<|A|/2$. Let $T_1=\{e\in A: S\cup\{e\}\in\cI_1\}$. Since $A\in\cI_1$, by the
matroid exchange property, we have $|T_1|\geq |A|-|S|>|A|/2$.
Similarly, if $T_2=\{e\in A: S\cup\{e\}\in\cI_2\}$, then we have $|T_2|>|A|/2$. But since
$T_1, T_2\sse A$, this means that $T_1\cap T_2\neq\es$, and so if $e\in T_1\cap T_2$, then
$e$ can be added to $S$ while maintaining independence in both $\mat_1$ and $\mat_2$. This
contradicts the maximality of $S$.
\end{proof}

\section{Scheduling markets} \label{sec:sched} \label{sched}
Recall that here the agents are $n$ jobs that need to be assigned on $m$ machines. Each
job $j$ has a private strict total ordering over the machines, and a publicly-known
processing time $p_{ij}$ on machine $i$. 
An outcome is a partial assignment of jobs to machines, also called a schedule, that has
makespan at most a given value $T$. An agent prefers outcome $o$ to outcome $o'$ if he
prefers his assigned machine in $o$ to that in $o'$.

We obtain nearly tight results for scheduling markets.
Say that an algorithm is an $(\al,\beta)$-approximation if it always returns a
schedule with rank-approximation factor $\al$ and makespan at most $\beta T$.  
For parallel machines ($p_{ij}=p_j$ for all $i,j$), we give an 
$\bigl(O(\log n),O(1)\bigr)$-approximation, 
fully lex-truthfully (LT) implementable algorithm (Theorem~\ref{parallel-ub2}). 
We show that this bound is {\em tight} by proving an {\em algorithmic lower bound} showing
that every $(\al,\beta)$-approximation algorithm for parallel machines must have\\
$\al=\Omega(\max\{\log m,\log n\}/\beta)$ (Theorem~\ref{parallel-lb}). 
For the setting of general unrelated machines, we devise an 
$\bigl(O(\log n), O(1)\bigr)$-approximation algorithm (Theorem~\ref{unrelated}), 
however we do not know how to achieve this via a fully LT-implementable algorithm. We
leave this as an intriguing open question.

Let $N$ denote the set of jobs.
For $S\sse N$, let $\bsucc_S$ denote the restriction of $\bsucc$ to jobs in $S$, and
$\mrank_r(\bsucc_S)$ denote the maximum number of jobs from $S$ that can be assigned to
one of their top-$r$ machines with makespan at most $T$. 
Observe that
$\mrank_r(\bsucc_{S\cup T})\leq\mrank_r(\bsucc_S)+\mrank_r(\bsucc_T)$. 

\vspace{-1ex}
\paragraph{Parallel machines}
Our results rely on a bucketing argument coupled with Theorem~\ref{matroid-thm} for
matroid-markets and some insights from the matroid-intersection problem.
We divide the set $N$ of jobs into $k=O(\log n)$ disjoint
classes $N_0,N_1,\ldots,N_k$ such that jobs in each class have roughly the same processing
time. Set $N_0:=\{j: p_j\leq\frac{T}{n}\}$, and 
$N_\ell:=\{j: 2^{\ell-1}\cdot\frac{T}{n}<p_j\leq 2^\ell\cdot\frac{T}{n}\}$ for
$\ell=1,\ldots,k:=\ceil{\log_2 n}$. 
Note that if $j\notin\bigcup_{\ell=0}^k N_\ell$, then
$p_j>T$, so $j$ cannot be assigned to any machine in any outcome and is not counted in 
$\mrank_r(\bsucc)$ for any position $r$. We assume for notational convenience that $N$
does not contain any such job in the sequel.
It will be convenient to ensure that $|N_0|\geq 1$. So we remove some fixed job $a$ from 
the $N_\ell$ set containing it and add it to $N_0$. 

Obtaining a good rank-approximation for a class $N_\ell$, where $\ell\geq 1$, with
makespan $O(T)$ amounts to a matroid-market problem (in fact, a {\em $b$-matching}
problem) since the makespan bound can be encoded by 
the constraint that at most $\frac{n}{2^{\ell-1}}$ jobs are assigned to each  machine.  
Any feasible schedule for $N_\ell$ yields a feasible allocation for the
corresponding matroid-market problem. So Theorem~\ref{matroid-thm} 
yields a pseudomonotone
$(2,2)$-approximation algorithm $f_\ell$ for class $N_\ell$, and a mechanism $\cM_\ell^\ve$ 
that $\ve$-implements it, for all $\ve>0$. 

\begin{theorem} \label{parallel-ub}
One can obtain a deterministic fully LT-implementable 
$\bigl(O(1), O(\log n)\bigr)$-approximation algorithm for parallel-machine markets.  
\end{theorem}

\begin{proof}
On input $\bsucc$, we output the schedule obtained by concatenating the schedule where
all jobs in $N_0$ are assigned to their top machine, and all the $f_\ell(\bsucc_{N_\ell})$
schedules. Note that the $N_0$-schedule has makespan at most $2T$. The resulting schedule,
denoted $f(\bsucc)$, has makespan $O(T\log n)$ and rank-approximation factor 2 (since 
$\mrank_r(\bsucc_N)\leq\sum_{\ell=0}^k\mrank_r(\bsucc_{N_\ell})$).  
Fix $\ve>0$.
The jobs in $N_0$ clearly have no incentive to lie. 
It is easy to see then that $f$ is $\ve$-LT implemented by the mechanism
that outputs the $N_0$-schedule concatenated with the (random) schedules output by 
the $\cM_\ell^\ve$ mechanisms, where we couple the random choices of all the $\cM_\ell^\ve$
mechanisms (i.e., their decisions are based on the outcomes of the same random coins) so
that $\Pr[\exists\ \ell: \cM_\ell^\ve(\bsucc_{N_\ell})\neq f_\ell(\bsucc_{N_\ell})]\leq\ve$.
\end{proof}

\begin{theorem} \label{parallel-ub2}
There is a randomized fully LT-implementable 
$\bigl(O(\log n), O(1)\bigr)$-approximation algorithm for parallel-machine markets, where
the rank-approximation and makespan bounds hold with probability 1. 
\end{theorem}

\begin{proof}
Consider an input $\bsucc$.
As before, we assign all jobs in $N_0$ to their top machine. 
Note that simply picking a class $N_\ell$ with probability $\frac{1}{k}$ and outputting
the concatenation of the $N_0$-schedule and $f_\ell(\bsucc_{N_\ell})$ is not enough
since this only yields $O(k)$ rank approximation in expectation.
Instead, we build upon the above ideas and leverage some results about the
matroid-intersection problem.   

Consider the following bipartite graph representing the concatenation $\sg$ of all
the $f_\ell(\bsucc_{N_\ell})$ schedules. We have a node for every machine, and every job
not in $N_0$, and an edge $(i,j)$ if $j$ is assigned to machine $i$ in schedule $\sg$. Now
set $x_{ij}=\frac{1}{k}$ for every edge $(i,j)$. 
Define $A_{i,\ell}:=\ceil{\frac{n}{2^{\ell-1}k}}$ for all $i,\ell$ and 
$B_r:=\floor{\rank_r(\sg;\bsucc_{N\sm N_0})/k}$ for all $r$.
Consider the following polytope: 
\begin{equation}
\begin{split}
\cP:=\Bigl\{ y\in\R^{[m]\times(N\sm N_0)}: \quad 
\sum_{j\in N_\ell}y_{ij} & \leq A_{i,\ell} \quad \forall i\in[m], \ell\in[k], \\ 
\sum_{j:\pos{\succ_j}{\sg(j)}\leq r}y_{ij} & \geq B_r \quad \forall r\in[m],
\qquad 0\leq y_{ij}\leq 1 \quad \forall i\in[m],\ j\notin N_0\Bigr\}. 
\end{split}
\label{mint}
\end{equation}
We claim that $\cP$ has integral extreme points. 
Any extreme point of
$\cP$ is defined by a linearly independent system of tight constraints comprising some
$\sum_{j\in N_\ell}y_{ij}=A_{i,\ell}$ equalities whose supports are disjoint, and some
$\sum_{j:\pos{\succ_j}{\sg(j)}\leq r}y_{ij}=B_r,\ y_{ij}=1$ equalities whose supports form
a laminar family. The constraint matrix of such a system thus corresponds to equations
coming from two laminar set systems; such a matrix is known to be totally unimodular (TU)
(see, e.g.,~\cite{Schrijver03}), and hence a solution to this system is integral.

Note that $x\in\cP$, so it can be expressed as a convex combination of some extreme
points of $\cP$. Equivalently, $x$ yields a distribution over partial schedules for 
$N\sm N_0$. 
Let $Y$ be a random schedule, or equivalently vector in $\R^{[m]\times (N\sm N_0)}$,
sampled from this distribution. 
Note that $\Pr[\text{$j$ is assigned in $Y$}]=x_{ij}=\frac{1}{k}$ for $j\notin N_0$. 
The makespan of $Y$ is at most $6T$ with probability 1. This is because 
$\sum_{j}p_jY_{ij}
\leq\sum_{\ell=1}^k\bigl(1+\frac{n}{2^{\ell-1}k}\bigr)\cdot 2^\ell\cdot\frac{T}{n}
\leq 2^{k+1}\cdot\frac{T}{n}+2T\leq 6T$.
Let $\Pi$ be the (random) schedule obtained by concatenating the $N_0$-schedule with $Y$. 
Then $\Pi$ has makespan at most $8T$ with probability 1.
Also, $\rank_r(\Pi;\bsucc)\geq |N_0|+B_r$ with probability 1.
Now $B_r\geq\floor{\mrank_r(\bsucc_{N\sm N_0})/2k}$. 
Finally, 
\begin{equation*}
\begin{split}
\mrank_r(\bsucc) & \leq |N_0|+\mrank_r(\bsucc_{N\sm N_0}) \\
& \leq|N_0|+\max\{2k,4k\floor{\mrank_r(\bsucc_{N\sm N_0}/2k)}\} \\
& \leq 4k(|N_0|+B_r),
\end{split}
\end{equation*}
where the latter inequality follows since $|N_0|\geq 1$. 
Thus, the randomized algorithm $f$ that outputs the random schedule $\Pi$ is an 
$\bigl(O(\log n), O(1)\bigr)$-approximation with probability 1.

\medskip
We now proceed as in the proof of Theorems~\ref{pseudo} and~\ref{newpseudo} to devise a
mechanism $\cM$ that fully LT-implements $f$. 
Fix $\ve>0$, and $\ve_1>\ldots>\ve_m$ such that
$\sum_{r=1}^m\ve_r=\ve$. Consider input $\bsucc$.
Let $Y^\bsucc$ be the random schedule for $N\sm N_0$ for input $\bsucc_{N\sm N_0}$ as obtained
above. 
Mechanism $\cM$ always assigns jobs in $N_0$ to their top machines.
For jobs in $N\sm N_0$, it returns schedule $Y^\bsucc$ with probability $1-\ve$.
For each $j\notin N_0$ and $r\in[m]$, with probability $\frac{\ve_r}{n}$, it returns the
schedule where $j$ is assigned to its $r$-th ranked machine $\alt{\succ_j}{r}$, and all
other jobs are unassigned. Clearly, $\cM(\bsucc)=f(\bsucc)$ with probability at least $1-\ve$. 

Jobs in $N_0$ do not benefit by lying. Consider a job $j\in N_\ell$, where
$\ell\geq 1$. Let $\bsucc'=(\succ_j,\succ_{-j})$, where $\succ'_j\neq\succ_j$. 
Let $x_{ij}=x_{ij}(\bsucc)$ and $x'_{ij}=x_{ij}(\bsucc')$ denote the probabilities that
$j$ is assigned to $i$ under the random schedules $Y=Y^\bsucc$ and $Y'=Y^{\bsucc'}$
respectively. 
Then, 
\begin{equation*}
\begin{split}
\Dt_{ij} & := \Pr[\text{$j$ assigned to $i$ in $\cM(\bsucc)$}]
-\Pr[\text{$j$ assigned to $i$ in $\cM(\bsucc')$}] \\
& =
(1-\ve)(x_{ij}-x'_{ij})+\frac{1}{n}\cdot\bigl(\ve_{\pos{\succ_j}{i}}-\ve_{\pos{\succ'_j}{i}}\bigr).
\end{split}
\end{equation*}
Considering machines in the preference order of $\succ_j$, let $\hi$ be the first machine
such that $\pos{\succ_j}{\hi}\neq\pos{\succ'_j}{\hi}$. 
Then $\pos{\succ_j}{\hi}<\pos{\succ'_j}{\hi}$. 
If $x'_{ij}=0$ for all $i\succeq_j\hi$, then $\Dt_{ij}\geq 0$ for all $i\succeq_j\hi$, and  
$\Dt_{\hi j}>0$, so we are done. Otherwise, $j$ is assigned to some machine
$i'\succeq_j\hi$ in $f_\ell(\bsucc'_{N_\ell})$. Since all machines $i\succ_ji'$ have
$\pos{\succ_j}{i}=\pos{\succ'_j}{i}$ and $f_\ell$ is pseudomonotone, it must be that
$j$ is assigned to $i''\succeq_j i'$ in $f_\ell(\bsucc_{N_\ell})$. So
$x_{ij}=x'_{ij}$, and hence, $\Dt_{ij}=0$, for all $i\succ_j i''$. If $i''\neq i'$, then
$\Dt_{i''j}>0$, otherwise $\Dt_{ij}=0$ for all $i\succ_j\hi$ and $\Dt_{\hi j}>0$. 
Thus, $\cM$ is lex-truthful.
\end{proof}

\begin{theorem} \label{parallel-lb} 
There exists an instance of a parallel-machine market where any schedule with $\beta T$ 
makespan has rank-approximation factor $\Omega(\max\{\log m,\log n\}/\beta)$. 
\end{theorem}

\begin{proof}
We create an instance with $n=O(m\ln m)$ jobs as follows.
We create a set $A^{(1)}$ of $m$ jobs of size (i.e., $p_j$) $T$ partitioned into
$A^{(1)}_1,\ldots,A^{(1)}_m$, where each $A^{(1)}_i$ consists of a single job whose first
preference is machine $i$. 
We create a set $A^{(2)}$ of $2(m-1)$ jobs of size $\frac{T}{2}$ partitioned into
$A^{(2)}_2,\ldots,A^{(2)}_m$, all of which have machine 1 as their first preference. 
Each set $A^{(2)}_i$ has two jobs, both having machine $i$ as their 2nd preference.
In general for $i<k$, we have a set $A^{(i)}$ of $2^{i-1}(m-i+1)$ jobs of size
$\frac{T}{2^{i-1}}$ partitioned into $A^{(i)}_i,\ldots,A^{(i)}_m$, all of which have
machine $r$ as their $r$-th preference for $r=1,\ldots,i-1$. 
Each set $A^{(i)}_\ell$ has $2^{i-1}$ jobs, all of which have machine $\ell$ as their
$i$-th preference.  
Finally, we have a set $A^{(k)}$ of $2^km$ jobs of size $\frac{T}{2^k}$ partitioned into
$A^{(k)}_k,\ldots,A^{(k)}_m$, all having machine $r$ as their $r$-th preference for
$r=1,\ldots,k-1$. Each set $A^{(k)}_\ell$ has at least $2^k$ jobs, all of which have
machine $\ell$ as their $k$-th preference. 
The remaining preferences of the jobs play no role, and may be set arbitrarily.
Let $\bsucc$ be the resulting preference profile.

For $r\in[k]$, we have $\mrank_r(\bsucc)\geq 2^{r-1}(m-r+1)+2^k(r-1)\geq 2^{r-1}m$ obtained
by assigning all jobs in $A^{(r)}_\ell$ to machine $\ell$ for $\ell=r,\ldots,m$, and any
$2^k(r-1)$ jobs from $A^{(k)}$ to machines $1,\ldots,r-1$. Suppose we have a schedule
$\sg$ with makespan $\beta T$ that achieves $\al$ rank approximation.
Then, $\rank_r(\sg;\bsucc)\geq\frac{2^{r-1}m}{\al}$ for all $r=1,\ldots,k$. Let $s_r$ be the
number of jobs assigned to their $r$-th ranked machine in $\sg$, and $t_r$ be the number
of jobs of size at least $\frac{T}{2^{r-1}}$ assigned to their $r$-th ranked machine in
$\sg$. 
Observe that $t_r\geq s_r-\beta 2^k$ since the jobs counted in $s_r$ but not in $t_r$ lie
in $\bigcup_{\ell=r+1}^kA^{(\ell)}$, all of which have machine $r$ as their $r$-th ranked
machine; at most $\beta 2^k$ such jobs can be accommodated within makespan $\beta T$.
Now $\beta mT$ is at least the total size of all jobs scheduled by $\sg$, which is at
least $\sum_{r=1}^k(s_r-\beta 2^k)\cdot\frac{T}{2^{r-1}}
\geq\sum_{r=1}^k s_r\cdot\frac{T}{2^{r-1}}-\beta 2^{k+1}T$. So
$$
\beta(mT+2^{k+1}T)
\geq\frac{1}{2}\sum_{r=1}^ks_r\sum_{\ell=r}^k\frac{T}{2^{\ell-1}}
=\frac{1}{2}\sum_{\ell=1}^k\frac{T}{2^{\ell-1}}\sum_{r=1}^\ell s_r
=\frac{1}{2}\sum_{\ell=1}^k\frac{T}{2^{\ell-1}}\cdot\rank_\ell(\sg;\bsucc)
\geq\frac{1}{2}\sum_{\ell=1}^k\frac{T}{2^{\ell-1}}\cdot\frac{2^{\ell-1}m}{\al}.
$$
Taking $k=\log_2 m$, this gives $3\beta mT\geq\frac{kmT}{2\al}$, so
$\al\geq\frac{k}{6\beta}=\Omega(\log m/\beta)$. Also, the
number of jobs is at most $k\cdot 2^k=O(m\log m)$, so $\al$ is also $\Omega(\log n/\beta)$.
\end{proof}

\vspace{-1ex}
\paragraph{Unrelated machines}
We obtain an $\bigl(O(\log n), O(1)\bigr)$ approximation for the general setting of
unrelated machines. 

\begin{theorem} \label{unrelated} 
There is a deterministic $\bigl(O(\log n),O(1)\bigr)$ approximation algorithm for 
scheduling markets.
\end{theorem}

\begin{proof} 
We will need Lemma~\ref{rnkcomp} stated below.
Fix an input $\bsucc$.
We use a different kind of bucketing argument where we group ranks that have roughly the
same value of $\mrank_r(\bsucc)$. 
For $r\in[m]$, let $\sg^r$ be the schedule given by Lemma~\ref{rnkcomp} that yields a
2-approximation to $\mrank_r(\bsucc)$, $N_r$ be the set of jobs assigned by $\sg_r$, 
and $n_r=|N_r|$. We may assume that $n_1\leq n_2\leq\ldots n_m$. Define $n_0=0$.
If $n_m=0$, then $\mrank_r(\bsucc)=0$ for all $r\in[m]$, and we return the null
assignment. So assume otherwise in the sequel.

Define $r_0:=0<r_1<r_2<\ldots<r_k<r_{k+1}=m+1$ as follows: $r_\ell$ is the smallest $r$
such that $n_r>4n_{r_{\ell-1}}$ for $\ell=1,\ldots,k$, and $n_m\leq 4n_{r_k}$. Thus,
$k\leq\ceil{\log_4 n}$ and $n_{r_\ell}\leq n_r\leq 4n_{r_{\ell}}$ for all
$r\in[r_{\ell},r_{\ell+1})$ and all $\ell=0,\ldots,k$.
For $\ell\in[k]$, define $S_{r_\ell}:=N_{r_\ell}\sm(\bigcup_{q=1}^{\ell-1}N_{r_q})$; note
that $|S_{r_\ell}|\geq \frac{2n_{r_\ell}}{3}$.

If $n_{r_k}<2k$ for all $r$, we simply return the assignment $\sg^{r_1}$. Clearly, this
yields a $2k$ rank approximation.
Otherwise, let $q$ be the smallest index $\ell$ such that $n_{r_\ell}\geq 2k$. 
Let $S=\bigcup_{\ell=q}^kS_{r_\ell}$.
Let $\sg$ be the schedule for $S$, where each job $j\in S_{r_\ell}$ is assigned to the
machine $\sg^{r_\ell}(j)$, for $\ell=q,\ldots,k$. 
Let $L_i:=|\{j: \sg(j)=i\}|$.
Consider the following bipartite graph, which is similar to the bipartite graph
constructed in the GAP-rounding algorithm~\cite{ST93}. 
We have a node for every job in $S$, and a node $(i,c)$ for every machine $i$ and
$c=1,\ldots,\ceil{\frac{L_i}{k-q+1}}$. 
We sort the jobs assigned to $i$ in $\sg$ in non-increasing $p_{ij}$ order (breaking ties
arbitrarily), and create an edge $\bigl((i,c),j\bigr)$ if $\sg(j)=i$ and its position in
this ordering lies in $\{(c-1)(k-q+1)+1,\ldots,c(k-q+1)\}$.
Let $E$ be the edge-set of this bipartite graph.
Consider the following polytope:
\begin{equation}
\begin{split}
\cQ:=\Bigl\{ y\in\R^{E}: & \quad
\sum_{j:((i,c),j)\in E}y_{(i,c),j}\leq 1 \quad \forall i\in[m],
c=1,\ldots,\ceil{\tfrac{L_i}{k-q+1}}, \\ 
\sum_{((i,c),j)\in E: j\in S_{r_\ell}}y_{(i,c),j} & \geq\floor{\frac{|S_{r_\ell}|}{k-q+1}} 
\quad \forall \ell=q,\ldots,k,
\qquad 0\leq y_{(i,c),j}\leq 1 \quad \forall \bigl((i,c),j\bigr)\in E\Bigr\}. 
\end{split}
\label{qint}
\end{equation}
As with the polytope $\cP$ (see \eqref{mint}), the constraint-matrix
defining an extreme point of $\cQ$ corresponds to equations coming from two laminar
systems, which is TU, so $\cQ$ has integral extreme points. 
Setting $x_{(i,c),j}=\frac{1}{k-q+1}$ for every edge $\bigl((i,c),j\bigr)$, note that
$x\in\cQ$. So we can find an integral $y\in\cQ$, which we interchangeably view as a
partial assignment of $S$. We return the schedule $\pi$ obtained by concatenating
$\sg^{r_1}$ with this assignment $y$. 

The schedule $\sg^{r_1}$ has makespan at most $T$. By the standard GAP-rounding proof 
in~\cite{ST93}, the makespan of $y$ is at most 
$$
T+\tfrac{1}{k-q+1}\cdot\sum_{j: \sg(j)=i}p_{ij}
=T+\tfrac{1}{k-q+1}\cdot\sum_{\ell=q}^k\sum_{j\in S_{r_\ell}:\sg^{r_\ell}(j)=i}p_{ij}\leq
2T.
$$
So $\pi$ has makespan at most $3T$.
Fix some rank $r$. If $r<r_1$, then $\mrank_r(\bsucc)=0$.
If $r_1\leq r<r_q$, we have 
$\rank_r(\pi;\bsucc)\geq n_{r_1}\geq 1>n_r/2k\geq\mrank_r(\bsucc)/4k$.
Otherwise, suppose $r\in[r_\ell,r_{\ell+1})$, where $\ell\geq q$.
Then $\rank_r(\pi;\bsucc)\geq\floor{\frac{|S_{r_\ell}|}{k}}
\geq\floor{\frac{2n_{r_\ell}}{3k}}\geq\frac{n_{r_\ell}}{3k}$, where the last inequality
follows since $n_{r_\ell}\geq n_{r_q}\geq 2k$, and
$\frac{n_{r_\ell}}{3k}\geq\frac{n_r}{12k}\geq\frac{\mrank_r(\bsucc)}{24k}$. So $\pi$ has
$O(k)$ rank approximation.
\end{proof}

\begin{lemma} \label{rnkcomp}
For any preference-profile $\bsucc$, any set $S\sse N$, and any rank $r$, one can
efficiently compute a 2-approximation to $\mrank_r(\bsucc_S)$.
\end{lemma}

\begin{proof} 
Shmoys and Tardos~\cite{ST93} proved the following result about GAP. Let $\{c_{ij}\}$
be ``assignment costs'' for assigning jobs to machines, 
{\em which could also be negative}. Consider the following LP, where $i$ indexes the
machines and $j$ indexes the jobs.  
\begin{alignat}{3}
\min && \quad \sum_{i,j} c_{ij}x_{ij} & \tag{P} \label{primal} \\
\text{s.t.} && \sum_i x_{ij} & = 1 \qquad && \forall j \label{asgn} \\
&& \sum_j t_{ij}x_{ij} & \leq L_i && \forall i \notag \\
&& x_{ij} & \geq 0  && \forall i,j \notag \\
&& x_{ij} & = 0 && \forall i, j\text{ s.t. }t_{ij}>L_i. \notag
\end{alignat}
\cite{ST93} showed that a fractional solution $x$ to \eqref{primal} can be rounded to
an integer solution $\tx$ of cost at most the cost of $x$ such that the total load 
$\sum_j t_{ij}\tx_{ij}$ on each machine $i$ is at most $2L_i$. Examining their rounding
algorithm more closely, one can infer that the total load on each machine $i$ under $\tx$
is at most $L_i$ if $\sum_j x_{ij}\leq 1$, and at most $L_i+\max_{j:x_{ij}>0}t_{ij}$
otherwise.

We apply this result to our problem of approximating $\mrank_r(\bsucc_S)$ as follows. 
The set of jobs is $S$. 
Our problem is a prize-collecting problem, which we can reduce to GAP by creating a
``machine'' $I_j$ for every job $j\in S$ and setting $t_{I_jk}=0$ if $k=j$ and $\infty$
otherwise. There is no makespan bound for these $I_j$ machines. For every (regular)
machine $i$ and job $j$, we set $t_{ij}=p_{ij}$ if $\pos{\succ_j}{i}\leq r$ and $\infty$
otherwise; the makespan bound for $i$ is $T$. Finally, our objective is to maximize the
number of jobs assigned to the regular machines (with $t_{ij}<\infty$). In terms of the LP 
\eqref{primal}, constraint \eqref{asgn} now reads $\sum_i x_{ij}+z_j=1$ for every $j\in
S$, where $z_j$ indicates if $j$ is assigned to $I_j$, and the objective function is to
maximize $\sum_{i}\sum_{j\in S}x_{ij}$.

Applying the GAP-rounding algorithm, we obtain an assignment $\tx$ with makespan at most
$2T$ such that $\rank_r(\tx;\bsucc_S)=\mrank_r(\bsucc_S)$. To turn this into a feasible
schedule with makespan $T$, we leverage the stronger property of the rounding algorithm
mentioned above. If the load on machine $i$ under $\tx$ is more than $T$, then we know
that $i$ has at least two jobs assigned to it, and there is a job assigned to $i$ whose
removal decreases the load on $i$ to at most $T$. We simply remove this job from every
overloaded machine $i$. This reduces the number of jobs assigned to an overloaded machine
$i$ by a factor of at most 2 (since $\sum_j\tx_{ij}\geq 2$), so now we obtain a schedule
with makespan $T$ where the number of jobs assigned (to one of their top-$r$ machines) is
at least $\mrank_r(\bsucc_S)/2$. 
\end{proof}

\section{Mechanisms for general ordinal settings} \label{general}
In this section, we evaluate the strength and flexibility provided by the notions of rank
approximation and lex-truthfulness in general ordinal settings. 
We devise an $O(\log n)$-rank-approximation randomized mechanism, and show that this
guarantee is tight for randomized mechanisms (Theorems~\ref{randrank-ub} and~\ref{randrank-lb}). 
We also observe that deterministic mechanisms cannot in general achieve
good rank approximation.  
Next, we consider lex-truthfulness and justify our earlier remark that lex-truthfulness
allows one to circumvent Gibbard's impossibility result. 
We describe a rich class of pseudomonotone SCFs called {\em top-choice SCFs}, which thus
lead to (non-unilateral, non-duple) LT mechanisms.  

\vspace{-1ex}
\paragraph{Rank approximation}
It is easy to see that any deterministic dictatorial SCF
has rank approximation (at most) $n$. 
Also, the {\em plurality scoring rule} $f^\pl$, which returns the outcome that maximizes the 
number of agents who have it as their top choice, has
$\rank_1\bigl(f^\pl(\bsucc);\bsucc)\geq\frac{n}{m}$, so its rank-approximation factor
is at most $m$. 
It is not hard to prove a matching lower bound for deterministic mechanisms. 

\begin{theorem} \label{detmech-lb}
No deterministic mechanism can have rank approximation better factor than $\min\{n,m-1\}$
in general ordinal settings.
\end{theorem}

\begin{proof}
Consider a preference profile with $n$ agents and $n+1$ outcomes, where the top choices of
all agents are the distinct outcomes $\{1,\ldots,n\}$, while the second choice of all
agents is $n+1$. 
\end{proof}

Randomization leads to an exponential improvement (but no more), but we do not 
know how to achieve this in a lex-truthful manner. 

\begin{theorem} \label{randrank-ub} 
There is a randomized $O(\log n)$-rank approximation mechanism for general ordinal
settings. 
\end{theorem}

\begin{proof}
We first describe the mechanism, and then analyze its rank approximation. Fix a preference
profile $\bsucc$. For brevity, let $n_r=\maxrank_r(\bsucc)$.  
Let $o^*_r$ be an outcome with $\rank_r(o^*_r;\bsucc)=n_r$. 
We use a bucketing argument where we group ranks that have roughly the same $n_r$ value. 
Define $r_1:=1<r_2<\ldots<r_k<r_{k+1}:=m+1$ be such that 
$n_{r_\ell}\leq n_r\leq 2n_{r_\ell} \quad \forall r\in[r_\ell,r_{\ell + 1})\cap\Z$, 
for all $\ell=1,\ldots,k$
Observe that $k\leq\ceil{\log_2 n}$. The mechanism chooses an index $\ell \in [k]$ 
uniformly at random, and outputs $o^*_{r_\ell}$. 

To argue about the rank approximation, consider any rank $r$. 
Suppose $r\in[r_\ell,r_{\ell+1})$. If we choose index $\ell$, which happens with
probability $1/k$, then at least 
$\rank_r(o^*_{r_\ell};\bsucc)\geq\rank_{r_\ell}(o^*_{r_\ell};\bsucc)=n_{r_\ell}\geq\frac{n_r}{2}$
agents are allotted a top-$r$ item.  
So $\Exp[\rank_r(\cM(\bsucc);\bsucc)]\geq\frac{n_r}{2k}$.
\end{proof}

\begin{theorem} \label{randrank-lb} 
Every randomized mechanism has rank-approximation factor $\Omega(\log n)$.
\end{theorem}

\begin{proof}
Fix a parameter $k$. We construct an instance with $n = 2^{k+1} - 2$ agents and $m =
(k-1)\cdot(2^{k+1}-2) + k$ outcomes. 
The agents are divided into $k$ groups $A_1,\ldots,A_k$, where $|A_\ell| = 2^\ell$. There
are $k$ special outcomes $\{o_1,\ldots,o_k\}$. 
The remaining $m-k$ outcomes are partitioned into $n$ groups $O_1,\ldots,O_n$, each having $k-1$
outcomes. 
We now describe the preference lists. For every agent $j$ in group $A_\ell$, outcome
$o_\ell$ is their $\ell$-th ranked outcome, and the outcomes in $O_j$ occupy the other
positions in $[k]\sm\{\ell\}$; the exact positions of these outcomes are irrelevant. 
The outcomes in positions $r\geq k+1$ are also immaterial. 
Thus, the top-$k$ outcome sets of agents $j$ and $j'$ are: disjoint if they are from
different groups, and have exactly one outcome, $o_\ell$, in common, at the $\ell$-th 
position, if they both belong to group $A_\ell$. 
Let $\bsucc$ denote this input.

Observe that $\maxrank_r(\bsucc) = 2^r$ for all $r\in[k]$, and the outcome achieving this
is $o_r$. Furthermore, $\rank_r(o)$ is $2^\ell$ if $o=o_\ell$ for $\ell\in[r]$, and is 
at most 1 otherwise. 

Now consider a randomized mechanism that attains rank approximation $\alpha$.
Let $p_\ell$ be the probability with which it returns the outcome $o_\ell$. Let $q$ be
the probability with which it returns an outcome in $\bigcup_{j=1}^nO_j$. 
Then, by the definition of rank approximation we have  
$q+\sum_{\ell>r}p_r + \sum_{\ell\leq r} p_\ell\cdot 2^\ell \geq \alpha\cdot 2^r$
for all $r\in[k]$.
Dividing this inequality by $2^r$ and summing over all $r=1,\ldots,k$, 
we obtain that 
$k\alpha\leq q\cdot\sum_{r=1}^k \left(\frac{1}{2^r}\right)
+\sum_{\ell=1}^k
p_\ell\cdot\left(\sum_{r<\ell}\frac{1}{2^r}+\sum_{r\geq\ell}\frac{2^\ell}{2^r}\right)
\leq q\cdot 1+\sum_{\ell=1}^k p_\ell\cdot 3\leq 3$.
Hence, $\al\leq\frac{3}{k}$.
\end{proof}

\vspace{-1.5ex}
\paragraph{Lex-truthful mechanisms}
Consider any SCF of the form $f(\bsucc)=g\bigl(\{\alt{\succ_j}{1}\}_{j=1}^n\bigr)$, where
\mbox{$g:O^n\mapsto O$} has the following property: for all\\
$o_{-j}=(o_1,\ldots,o_{j-1},o_{j+1},\ldots,o_n)\in O^{n-1}$ and all $o\in O$,
if $g(o,o_{-j})=o'$ then $g(o',o_{-j})=o'$.
We call such an SCF a {\em top-choice SCF} since it only looks at the top
choices of the players. It is not hard to see that the plurality scoring rule $f^\pl$
mentioned earlier (with a fixed tie-breaking rule for outcomes) is an example of such an
SCF. 
We show that {\em any} top-choice SCF is pseudomonotone, and so by Theorem~\ref{pseudo} is 
fully LT-implementable.  

\begin{theorem} \label{genlt}
Every top-choice SCF is pseudomonotone, and hence is fully LT-implementable.
\end{theorem}

\begin{proof}
Let $f$ be a top-choice SCF defined by $g:O^n\mapsto O$ having the required
property. Consider an agent $j$, and $\bsucc=(\succ_j,\bsucc_{-j})$, 
$\bsucc'=(\succ'_j,\bsucc_{-j})$. Let $o=\alt{\succ_j}{1}$. If $f(\bsucc)=o$ or
$f(\bsucc)=$\mbox{$f(\bsucc')$}, then we are done. 
Otherwise, since $f(\bsucc)\neq o$, we also have
$f(\bsucc')\neq o$ due to the property of $g$, and also $\pos{\succ'_j}{o}>1$ (otherwise
$f(\bsucc)=f(\bsucc')$), and so the pseudomonotonicity condition
(Definition~\ref{pseudodefn}) is satisfied.    
\end{proof}

\appendix

\section{Proofs omitted from Section 3} \label{app:ranklex} \label{append-ranklex}

\begin{proofof}{Theorem \ref{thm:hier}} 
Clearly $\univt\sse\strongt$.
If $\bp$ stochastically dominates $\bq$, then $\bp$ lex-dominates $\bq$, so
$\strongt\sse\lext$. 
If $\bp$ lex-dominates $\bq$, then $\bq$ cannot stochastically dominate $\bp$, so
$\lext\sse\weakt$. We now prove that the various inclusions are strict. 

\medskip\noindent
{\boldmath $\univt\subsetneq\strongt$}. \quad 
Fix a player $j$. Consider the unilateral mechanism
$\cM$ that returns one of the top 2 outcomes of $j$, each with probability
$\frac{1}{2}$. $\cM$ is clearly strongly truthful. But it is not universally truthful. 
Consider some input $\bsucc=(\succ_j,\succ_{-j})$.
If $\cM$ is a mixture of deterministic truthful mechanisms, then this mixture must assign
a probability mass exactly $\frac{1}{2}$ to deterministic truthful mechanisms $\cM_1$
satisfying $\cM_1(\bsucc)=o=\alt{\succ_j}{1}$; call these type-1 mechanisms. 
Similarly, it must assign probability mass exactly $\frac{1}{2}$ to deterministic truthful
mechanisms $\cM_2$ satisfying $\cM_2(\bsucc)=o'=\alt{\succ_j}{2}$; call these type-2
mechanisms.   

For any preference list $\succ^*_j$, a type-2 mechanism cannot return $o$ under the input
$(\succ^*_j,\bsucc_{-j})$ due to truthfulness, otherwise on input $\bsucc$, $j$ has
an incentive to lie in the type-2 mechanism and report $\succ^*_j$. 
Hence, for any preference list $\succ^*_j$, where $o$ is one of the top two outcomes, 
{\em every} type-1 mechanism must return $o$ on input $(\succ^*_j,\bsucc_{-j})$. A
symmetric argument shows that for any preference list $\succ^*_j$ where $o'$ is one of the 
top two outcomes, every type-2 mechanism must return $o'$ on input
$(\succ^*_j,\bsucc_{-j})$. 

Now consider some $\succ'_j$, where the top two outcomes are 
$o'', \ho\notin\{o,o'\}$. Applying the 
arguments above we obtain that there are type-3 and type-4 deterministic truthful
mechanisms, both of which are assigned probability mass $\frac{1}{2}$ (in the mixture
yielding $\cM$): the type-3 mechanisms which always return $o''$ whenever $j$'s preference
list has $o''$ as one of the top two outcomes, and the type-4 mechanisms always return
$\ho$ whenever $j$'s preference list has $\ho$ as one of the top two outcomes.

Now some mechanism $\cM'$ in the mixture yielding $\cM$, must be of multiple types, say
type-1 and type-3 for illustration. Then, if $\succ''_j$ has $o$ and $o''$ as the top two 
outcomes of $j$, $\cM'$ must return both $o$ and $o''$ on input $(\succ''_j,\bsucc_{-j})$,
which cannot happen.

\medskip\noindent
{\boldmath $\strongt\subsetneq\lext$}. \quad 
Consider a setting with one player and three
outcomes: $a$, $b$, $c$. Consider the top-choice SCF $f$ defined by the following
function: $g(a)=a,\ g(b)=g(c)=c$, which satisfies the property required for $f$ to be a
top-choice SCF. By Theorem~\ref{genlt}, $f$ is pseudomonotone. Let $\cM$ be the LT
mechanism that $\frac{1}{3}$-implements $f$. 
Let $\succ=(b,a,c)$ denoting that $b$ is top-outcome, and $\succ'=(a,b,c)$.
Let $\bp=\cM(\succ)$ and $\bq=\cM(\succ')$. 
Then $\bp(b)+\bp(a)\leq\frac{1}{3}$ since $f(\succ)=g(b)=c$, but
$\bq(b)+\bq(a)\geq\frac{2}{3}$ since $f(\succ')=g(a)=a$. Thus, $\cM$ is not strongly
truthful. 

\medskip\noindent
{\boldmath $\lext\subsetneq\weakt$}. \quad 
Consider a setting with one player and four
outcomes: $a$, $b$, $c$, $d$. Let $\succ^*=(a,b,c,d)$, denoting that $a$ is the top 
outcome. Define the following randomized mechanism $\cM$: $\cM(\succ^*)$ returns $a$ with  
probability $\frac{1}{2}$ and $b$, $c$, $d$ with probability $\frac{1}{6}$; on every other 
input $\succ$, $\cM$ returns one of the top three outcomes of $\succ$ with probability
$\frac{1}{3}$. $\cM$ is weakly truthful, because if $\succ\neq\succ^*$ then $\cM$ assigns
total probability 1 to the top three outcomes of $\succ$. If $\succ=\succ^*$, then $\cM$
assigns probability $\frac{1}{2}$ to the top outcome $a$ under $\succ^*$, whereas for
every other input $\cM$ assigns probability at most $\frac{1}{3}$ to $a$. 

But $\cM$ is not lex-truthful: if $\succ=(a,b,d,c)$, then by reporting $(a,b,c,d)$, the
player can increase the probability of his top-outcome $a$ from $\frac{1}{3}$ to
$\frac{1}{2}$. 
\end{proofof}

\begin{proofof}{Theorem~\ref{newpseudo}}
We mimic the proof of Theorem~\ref{pseudo}. 
For all $j$, and all $r\in[m_j]$, fix some outcome $o^j_{r}\in O^j_r$ form the
indifference class $O^j_r$ of agent $j$.  

We prove part (i) first.
Our randomized mechanism $\cM$ does the following. On input
$\bsucceq$, it returns $f(\bsucceq)$ with probability $(1-\eps)$; with probability $\eps$,
it picks a random agent $a$ and returns outcome $o^a_r$ with probability $\eps^a_r/\eps$, 
where $\eps^a_1>\cdots >\eps^a_{m_a}>0$ are such that $\sum_{r=1}^{m_a}\eps^a_r=\eps$.   

Clearly, $\cM$ $\eps$-implements $f$. To prove lex-truthfulness, fix
an agent $j$ and consider any $\bsucceq'=(\succeq'_j,\bsucceq_{-j})$, where 
$\succeq'_j \neq\succeq_j$. 
Let $o=f(\bsucceq)$ and let $o'=f(\bsucceq')$. 
Let $O^j_{t_1}$ and $O^j_{t_2}$ be the indifference classes of $j$ containing outcomes $o$
and $o'$ respectively. 
Also, let $\bp=\cM(\bsucceq),\ \bq=\cM(\bsucceq')$.  

Considering indifference classes in the preference order of $\succeq_j$,
let $O^j_r$ be the first indifference class such that
$\pos{\succeq_j}{o^j_r}<\pos{\succeq'_j}{o^j_r}$.  
Let $o''=o^j_r$.
By pseudomonotonicity of $f$, we know that $o\succeq_j o'$ or $o''\succ_j o'$. In the
latter case, we have $\bp(O^j_t)-\bq(O^j_t)\geq 0$ for all $t$ such that 
$o^j_t\succeq_j o''$, and $\bp(O^j_r)-\bq(O^j_r)>0$, so we are done. 

If $o\succeq_j o'$, and $O^j_{t_1}=O^j_{t_2}$ or $o''\succeq_j o$, then the above argument
still holds. So suppose $o\succ_j o'$ and $o\succ_j o''$.  
Then $\bp(O^j_t)-\bq(O^j_t)\geq 0$ for all $t$ such that $o^j_t\succeq_j o$ and
$\bp(O^j_{t_1})-\bq(O^j_{t_1})>0$, so again we are done. 

\medskip
Now consider part (ii).
Let $\cM$ be an LT mechanism that $\ve$-LT implements $f$, where $\ve<\frac{1}{2}$. 
Suppose for a contradiction, there is some agent $j$, and
$\bsucceq=(\succeq_j,\bsucceq_{-j})$ and $\bsucceq'=(\succ'_j,\bsucceq_{-j})$ such that
$o=f(\bsucc)$ and $o'=f(\bsucc')$ violate the conditions for pseudomonotonicity. 
Then we have $o'\succ_j o$ and for every outcome
$o''\succ_j o'$, we have $\pos{\succ_j}{o''}=\pos{\succ'_j}{o''}$. 
Let $O^j_{t_1}$ and $O^j_{t_2}$ be the indifference classes of $j$ containing outcomes $o$
and $o'$ respectively. 
Let $\bp=\cM(\bsucceq)$ and $\bq=\cM(\bsucceq')$.

Since $\cM$ $\ve$-LT implements $f$, we have
$\bp(O^j_{t_2})\leq\ve$ and
$\bq(O^j_{t_2})\geq\bq(o')\geq 1-\ve$, so $\bp(O^j_{t_2})<\bq(O^j_{t_2})$. 
Let $O^j_{r_1},\ldots,O^j_{r_\ell}$ be the indifference classes of $j$ that are ranked
higher than $O^j_{t_2}$ under $\succeq_j$, ordered so that 
$o^j_{r_1}\succ_j o^j_{r_2}\succ_j\cdots\succ_j o^j_{r_\ell}$.
Since $\pos{\succ_j}{o''}=\pos{\succ'_j}{o''}$ for all $o''\succ_j o'$, 
$O^j_{r_1},\ldots,O^j_{r_\ell}$ are also the indifference class of $j$ that are ranked
higher than $O^j_{t_2}$ under $\succeq'_j$, and 
we have $o^j_{r_1}\succ'_j o^j_{r_2}\succ'_j\cdots\succ'_j o^j_{r_\ell}$.
As in the proof of Theorem~\ref{pseudo}, this implies that 
$\bp(O^j_{r_t})=\bq(O^j_{r_t})$ for all $t=1,\ldots,\ell$, which contradicts the fact that
$\cM$ is lex-truthful. 
\end{proofof}

\section{Quality of known mechanisms for matching markets} \label{known-mat} \label{app:known-mat}
In this section, we investigate the rank approximation and lex-truthfulness of three
extensively studied mechanisms for matching markets. These are 
{\em random serial dictatorship} mechanism (RSD), Gale's {\em top-trading-cycle algorithm}
(TTCA), and the {\em probabilistic serial} mechanism (PS).  

\paragraph{Random Serial Dictatorship} Initially all items are marked unallocated. A
random permutation of agents is sampled and the agents are considered according to this
order. Each agent is allocated his best item among the unallocated items. This item
henceforth is marked allocated. 

\paragraph{Top Trading Cycle} This appears in a paper by Shapley and Shubik~\cite{SS74}
who attributed it to David Gale and is applicable when the number of items equals the
number of agents.  

The algorithm starts with an arbitrary assignment $\sigma$ of agents to items. This
assignment, which is called the initial endowment of agents,  is independent of the
preference orders of the agents. Subsequently, the agents will {\em trade} among
themselves to return the final allocation.

The algorithm then proceeds in rounds. Initially all agents are marked active.
In each round, one constructs a directed graph with the active agents as nodes. There is
an arc from agent $j$ to agent $j'$, if the item $\sigma(j')$ is the top choice of agent
$j$ among the items owned by the active agents, that is, the set $\{\sigma(j): j \textrm{
  active }\}$. Note that each agent has out-degree exactly $1$ (self loops are allowed and
counted as both out and in degree). Therefore, there exists at least one directed cycle in
the graph. A cycle (self loops are also cycles) is picked arbitrarily. For each arc
$(j,j')$ in the cycle, we allocate agents $j$ the item $\sigma(j')$. We mark all agents in
this cycle inactive and proceed to the next round. The algorithm stops when all agents are
marked inactive. 

\paragraph{Probabilistic Serial} This algorithm is due to  Bogomolnaia and
Moulin~\cite{BM01}. We describe the algorithm when the number of agents, $n$, equals the 
number of items, $m$.  

The algorithm first finds a {\em fractional matching}, that is, $x_{ij}$s for items $i$
and agents $j$ such that each $x_{ij} \geq 0$ and $\sum_{i\in[m]} x_{ij} = 1$ for all
agents $j$, and $\sum_{j\in[n]} x_{ij} = 1$ for all items $i$. By the Birkhoff-von Neumann
theorem, we can find a distribution on matchings such that the probability agent $j$ is
allocated item $i$ is exactly $x_{ij}$. This is the distribution returned by the
algorithm. 

The algorithm proceeds in rounds. Initially all $x_{ij}$'s are $0$. For any item $i$, we
denote its capacity as $\sum_{j\in[n]} x_{ij}$. Any item with capacity strictly less than
$1$ is called unallocated.  
In each round, every agent points to the best item among the unallocated items. For each
unallocated item $i$ we simultaneously raise $x_{ij}$ for all agents $j$ that point to
item $i$ at the {\em same} rate. This continues till some unallocated item's capacity
becomes $1$. At this point we end the round and proceed to the next round. The algorithm
terminates when all items are allocated. 
Since the procedure maintains that $\sum_{i\in[m]} x_{ij}$ is the same for all agents $j$,
at the end we end up with a fractional matching. 

\medskip 
A lot of literature exists on all three mechanisms; we point the reader to
surveys~\cite{SU08,AS10} for a detailed reference. Before stating the rank approximations
and lex truthfulness, let us mention some relevant known facts. RSD is strongly truthful
(in fact, it is universally truthful). 
TTCA is the only deterministic algorithm among the three. It is known that for {\em any}
initial endowment, the algorithm is truthful~\cite{SS74}.   
PS is known to be weakly truthful and not strongly truthful~\cite{BM01}. 
Bhalgat et al~\cite{BCK11} proved that the ordinal welfare factor (cf. \mySec{related}) of RS and PSD
are $1/2$, which is the best possible. The OWF of TTCA is $1/n$, as can be seen by
considering the input where all agents have the same preference list.

\paragraph{Rank Approximations of RSD, TTCA, and PS} 
We show that all three mechanisms have `bad' rank approximation. Rank approximation of TTCA
is at least $(n-1)$, while RSD and PS have rank approximation of
$\Omega(\sqrt{n})$. Recall that \mma has rank approximation $2$.

We know that TTCA is deterministic and truthful. It is also non-bossy; if an agent changes
his preference but still gets the same item, it implies that in the round when he gets
allocated an item, the cycle is the same as before, since no other changes
preferences. Therefore from Theorem~\ref{nobossy}, we get the rank approximation is at least
$(n-1)$. 

Consider an instance $\bsucc$ with $n$ agents and $n$ items with preference lists as
follows. Let $k = \ceil{\sqrt{n}}$. Agents $1\leq i\leq k$ have item $i$ as their top
choice. 
Agents $k+1\leq i\leq n$ have item $n$ as their top choice. These agents are now grouped
into $k$ groups $G_1,\ldots,G_k$, each group containing $n/k -1$ agents.  Agents in group
$G_\ell$ have item $\ell$ as their second choice. All the other choices of all agents is
immaterial and can be assumed to be arbitrary. Observe that $\maxrank_1(\bsucc) = k+1$. 

Let's first take RSD and calculate the expected number of agents who get their top
choice. With $1 - k/n$ probability, an agent $k+1\leq j \leq n$ shows up as the first
agent; he picks item $n$. No other agent $k+1\leq j\leq n$ gets his top choice. 
 Henceforth, for any $1\leq\ell\leq k$, the probability that a guy in $G_\ell$ shows up
 before agent $\ell$ is at least $1 - \frac{2k}{n}$. If that occurs, then agent $\ell$
 doesn't get his top choice. 
Therefore, the expected number of agents getting their top choice is at most $1 + 2k^2/n + o(n)$.
Thus, setting $k=\Theta(\sqrt{n})$, the rank approximation is $\Omega(\sqrt{n})$.

In PS, the calculation is easier. For $1\leq \ell \leq k$, we get $x_{\ell\ell} =
\frac{1}{n-k} + \frac{k}{n}\left(1 - \frac{1}{n-k}\right) = \frac{k-1}{n}$. For agent
$k+1\leq \ell \leq n$, we get $x_{n\ell} = \frac{1}{n-k}$. Therefore, the expected number
of agents getting their top choice in PS is precisely  
$1 + \frac{k(k-1)}{n}$. Setting $k=\Theta(\sqrt{n})$, we get that the rank approximation is $\Omega(\sqrt{n})$.

We do not know if the rank approximation for RSD and PS is $\Theta(\sqrt{n})$ or not.

\paragraph{Lex-Truthfulness of RSD, TTCA, and PS}
TTCA is truthful and RSD is universally truthful. Therefore, they are lex-truthful as well.
PS was shown to be weakly truthful by~\cite{BM01}. We show that in fact PS is lex-truthful as well.
The proof below is akin to the proof of weak truthfulness in~\cite{BM01} mentioned above; we include it for completeness.

\begin{theorem}
PS is lex-truthful.
\end{theorem}
\begin{proof}
Consider any preference profile $\bsucc$. By renaming items we may assume $\succ_j = (1,2,\ldots,n)$ for some agent $j$. Suppose agent $j$ misreports his preference as $\succ'_j \neq \succ_j$, and let $\bsucc' := (\succ'_j,\bsucc_{-j})$.
Let $k$ be the first position at which $\succ_j$ and $\succ'_j$ differ.
That is, for $r < k$, $\alt{\succ'_j}{r} = \alt{\succ_j}{r} = r$. Note that $j$ has `demoted' $k$ in the misreported preference, that is, $\pos{\succ'_j}{k} > k$.
Let $\bp$ and $\bq$ be the distributions over items that $j$ obtains on reporting $\succ_j$ and $\succ'_j$ respectively. Let $x$ and $x'$ be the respective fractional matchings.

Observe that since PS has a notion of time (since $x_{ij}$'s are incremented at a certain rate), we can define $x(t)$ as the assignment at time $t$. So $x(0) \equiv 0$. Let $t_0\geq 0$ be the time till which we have $x(t_0) \equiv x'(t_0)$. If $t_0$ is ill defined, then $x\equiv x'$ and so $\bp\equiv\bq$ and there's nothing to prove.
We must have that till time $t_0$, agent $j$ points to the same items in both runs, and right after that instant agent $j$ points to different items in the two runs. Say at $t_0$, agent $j$ pointed to item $k$ in the original run, and $k'$ in the new run. 
Observe  that all items $r < k$ have been completely allocated in both runs since $j$ is pointing to $k$ in the original run. Thus, $\bp(r) = \bq(r)$ for $r < k$ since $x(t_0) \equiv x'(t_0)$.

We claim $\bp(k) > \bq(k)$. This will show $\bp$ lexicographically dominates $\bq$.
To do so, we need to introduce some notation. Let $t^*$ and $t'$ be the times at which $k$ is completely allocated in the original and new run respectively. Let $t_1$ be the time at which $j$ points to $k$ in the new run. Observe $t_0 < t_1 \leq t'$. Also observe that $x_{jk} = t^* - t_0$ and $x'_{jk} = t' -t_1$.

Now, if $t' \leq t^*$, we get $x_{jk} > x'_{jk}$, and we are done. So we may assume $t' > t^*$.
For $t\geq t_0$, let $S(t,k)$ and $S'(t,k)$ be the set of agents pointing to item $k$ at time $t$. 
Observe that PS satisfies the following monotonicity condition: if an agent points to an item at time $t$, then he continues to do so till the item is fully allocated.
Using this, one can prove the following claim; we defer the proof to the end.

\begin{claim}\label{clm:mono-ps} 
For all $t_0 \leq t < t_1$, $|S'(t,k)| \geq |S(t,k)| - 1$, for $t_1 \leq t < t^*$, $|S'(t,k)| \geq |S(t,k)|$, and for $t^* \leq t < t'$, $|S'(t,k)| \geq |S(t^*,k)|$.
\end{claim}

Using the claim, we now show $x_{jk} > x'_{jk}$. Let $C$ denote the capacity of item $k$
at time $t_0$. We know that $C < 1$. Now, from the run of PS we get
\begin{equation}
\int_{t_0}^{t^*} |S(t,k)|dt \quad = \quad \left(1 - C\right) \quad = \int_{t_0}^{t'} |S'(t,k)|dt
\end{equation}

Using the claim above and rearranging, we get 
$$t_0 - t_1 \leq - \int_{t^*}^{t'} |S(t^*,k)|dt$$

Now suppose $|S(t^*,k)| = 1$, that is, in the original run only one guy points to item $k$. This must be agent $j$. This implies $|S(t,k)|=0$ for $t<t_0$, the time at which $j$ points to $k$. In particular, we get $C=0$, and thus $x_{jk} = 1$. We know that $x'_{jk'} > 0$ since $j$ points to $k'\neq k$ in the new run. Therefore, $x'_{jk} < 1$ since $\sum_{k\in I} x_{jk} = 1$. Thus, we may assume $|S(t^*,k)| > 1$, which implies that $t_0 - t_1 < - (t' - t^*)$. Thus, 
\begin{equation*}
x'_{jk} = t'-t_1 < t^* - t_0 = x_{jk}. \qedhere
\end{equation*}
\end{proof}

\begin{proofof}{Claim \ref{clm:mono-ps}} 
In fact, we claim that for every item $i \neq k$, the subset $S(t,i) \subseteq S'(t,i)$ for $t_0 \leq t < t'$. This can be proved by induction. Suppose the claim is true at some time; it is true at time $t_0$. The next interesting time $t$ is when some item is $i$ is completely allocated in one of the runs. By our assumption, this time $t$ occurs in the new run since $S'(t,i)\geq S(t,i)$ for $i\neq k$. 
At this point the agents pointing to $i$ point to different items increasing their corresponding $S'(t,\cdot)$s. The same occurs in the original run albeit at a later time say $t''$; however, by monotonicity property $S'(t'',i) \supseteq S'(t,i)$, and therefore $|S'(t'',i)| \geq |S(t'',i)|$.
For the item $k$, note that the above argument implies $|S'(t,i)\setminus k|\geq |S(t,i)\setminus k|$, and then after $t_1$, $j$ enters $S'(t,k)$ as well.
\end{proofof}

\end{document}